\documentclass{article}

\usepackage{url}
\usepackage{amsmath,amssymb}
\usepackage{amsthm}
\usepackage{tikz}
\usetikzlibrary{fit}
\usepackage{comment}
\usepackage[color=green!20]{todonotes}
\usepackage{xspace} 
\usepackage{xcolor}
\usepackage{fancybox}			
\usepackage{mathtools}	
\usepackage{booktabs}
\usepackage{scalefnt}
\usepackage{a4wide}	
\usepackage{enumerate} 

\newcommand{\graphm}{\textsc{Graph Motif}\xspace}

\newcommand{\pxtc}{\textsc{X3C}\xspace}

\newcommand{\np}{\mathsf{NP}}
\newcommand{\wone}{\mathsf{W[1]}}
\newcommand{\wtwo}{\mathsf{W[2]}}
\newcommand{\wpp}{\mathsf{W[P]}}
\newcommand{\fpt}{\mathsf{FPT}}
\newcommand{\xp}{\mathsf{XP}}

\newcommand{\p}{\mathsf{P}}

\newcommand{\bl}{block\xspace}
\newcommand{\bls}{blocks\xspace}

\newcommand{\C}{\mathcal{C}}

\newcommand{\I}{\mathcal{I}}

\newcommand{\R}{\mathcal{R}}
\renewcommand{\S}{\mathcal{S}}
\newcommand{\T}{\mathcal{T}}
\newcommand{\U}{\mathcal{U}}

\tikzstyle{edge} = [draw,-,rounded corners=8pt]
\tikzstyle{vertex}=[circle, draw, inner sep=2pt, minimum width=4pt]
\usetikzlibrary{decorations,decorations.pathmorphing,decorations.pathreplacing,patterns}

\renewcommand\leq\leqslant
\renewcommand\geq\geqslant

\newtheorem{theorem}{Theorem}

\newtheorem{lemma}[theorem]{Lemma}
\newtheorem{corollary}[theorem]{Corollary}

\newtheorem{definition}[theorem]{Definition}

\newtheorem{example}[theorem]{Example}

\newcommand{\Pb}[4]{%
\begin{center}
  \begin{tabular}{|l|}%
  \hline
    \begin{minipage}[c]{.95\textwidth}
      \smallskip%
      \par\noindent%
      #1
      \medskip%
      \par\noindent%
      $\bullet$
      \textbf{\textsf{Input}}: #2%
      \par\noindent%
      $\bullet$
      \textbf{\textsf{#4}}:
      #3%
      \par\noindent%
    \end{minipage}
  \\\hline
  \end{tabular}%
\end{center}
}%

\usepackage{etoolbox}

\newcounter{Bew1}
\newcounter{Bew2}
\newcounter{Def1}

\title{The Graph Motif problem parameterized by the structure of the input graph\footnote{An extended abstract of this work appears in~\cite{DBLP:conf/iwpec/BonnetS15}.}}

\author{\'Edouard Bonnet\footnote{\noindent Institute for Computer Science and Control, Hungarian Academy of Sciences, (MTA SZTAKI) \tt{bonnet.edouard@sztaki.mta.hu}} \and Florian Sikora\footnote{Universit\'{e} Paris-Dauphine, PSL Research University, CNRS, LAMSADE, Paris, France,  \tt{florian.sikora@dauphine.fr}}
}

\date{}

\begin{document}

\maketitle

\begin{abstract}
The \graphm problem was introduced in 2006 in the context of biological networks. 
It consists of deciding whether or not a multiset of colors occurs in a connected subgraph of a vertex-colored graph. 
\graphm has been mostly analyzed from the standpoint of parameterized complexity.
The main parameters which came into consideration were the size of the multiset and the number of colors.
In the many utilizations of \graphm, however, the input graph originates from real-life applications and has structure.
Motivated by this prosaic observation, we systematically study its complexity relatively to graph structural parameters. 
For a wide range of parameters, we give new or improved FPT algorithms, or show that the problem remains intractable. 
For the FPT cases, we also give some kernelization lower bounds as well as some ETH-based lower bounds on the worst case running time.
Interestingly, we establish that \graphm is $\wone$-hard (while in $\wpp$) for parameter max leaf number, which is, to the best of our knowledge, the first problem to behave this way.
\end{abstract}

\section{Introduction}

The \graphm problem has received a lot of attention during the last decade. 
Informally, \graphm is defined as follows: given a graph with arbitrary colors on the nodes and a multiset of colors called the motif, the goal is to decide if there exists a subset of vertices of the graph such that (1) the subgraph induced by this subset is connected and (2) the colors on the subset of vertices match the motif, i.e. each color appears the same number of times as in the motif. 
Originally, this problem is motivated by applications in biological network analysis~\cite{Lacroix2006}.
However, it also proves useful in social or technical networks~\cite{Betzler2011} or in the context of mass spectrometry~\cite{Bocker2009}. 

Studying biological networks allows a better characterization of species, by determining small recurring subnetworks, often called \emph{motifs}. 
Such motifs can correspond to a set of nodes realizing some function, which may have been evolutionary preserved. 
Thus, it is crucial to determine these motifs to identify common elements between species and transfer the biological knowledge. 
\graphm corresponds to topology-free queries and can be seen as a variant of a graph pattern matching problem with the sole topological requirement of connectedness. 
Such queries were also studied extensively for sequences during the last thirty years, and with the increase of knowledge about biological networks, it is relevant to extend these queries to networks~\cite{Pinter14}.

\section{Preliminaries and previous work}\label{sec:prelim}

For any two integers $x<y$, we set $[x,y]:=\{x,x+1,\ldots,y-1,y\}$, and for any positive integer $x$, $[x]:=[1,x]$.
If $G$ is a graph, we denote by $V(G)$ its set of vertices and by $E(G)$ its set of edges.
If $G=(V,E)$ is a graph and $S \subseteq V$, $E_G(S)$ denotes the subset of edges of $E$ having both endpoints in $S$.
If $G=(V,E)$ is a graph and $S \subseteq V$ is a subset of vertices, $G[S]$ denotes the subgraph of $G$ induced by $S$: $(S,E_G(S))$.
For a vertex $v \in V$, the set of neighbors of $v$ in $G$ is denoted by $N_G(v)$, and $N_G(S):=(\bigcup_{v \in S}N_G(v)) \setminus S$.
We define $N_G[v]:=N_G(v) \cup \{v\}$ and $N_G[S]:=N_G(S) \cup S$.
In all the previous definitions, we will lose the subscript $_G$ whenever the graph $G$ we are referring to is either implicit or irrelevant.
We say that a vertex $v$ \emph{dominates} a set of vertices $S$ if $S \subseteq N[v]$. 
A set of vertices $R$ \emph{dominates} another set of vertices $S$ if $S \subseteq N[R]$. 
If $G=(V,E)$ is a graph and $V' \subseteq V$, $G-V'$ denotes the graph $G[V \setminus V']$.
A \emph{universal vertex} $v$, in a graph $G=(V,E)$, is such that $N_G[v]=V$.
A \emph{matching} of a graph is a set of mutually disjoint edges.
In an explicitly bipartite graph $G=(V_1 \cup V_2,E)$, we call a matching of size $\min(|V_1|,|V_2|)$ a \emph{perfect matching}.
A \emph{cluster graph} (or simply, \emph{cluster}) is a disjoint union of cliques.
A \emph{co-cluster graph} (or, \emph{co-cluster}) is the complement graph of a cluster graph.
If $\mathcal C$ is a class of graphs, the \emph{distance to $\mathcal C$} of a graph $G$ is the minimum number of vertices to remove from $G$ to get a graph in $\mathcal C$.

If $f: A \rightarrow B$ is a function and $A' \subseteq A$, $f_{|A'}$ denotes the restriction of $f$ to $A'$, that is $f_{|A'}: A' \rightarrow B$ such that $\forall x \in A'$, $f_{|A'}(x):=f(x)$.
Similarly, if $E$ is a set of edges on vertices of $V$ and $V' \subseteq V$, $E_{|V'}$ is the subset of edges of $E$ having both endpoints in $V'$.


\medskip

\textbf{Multisets.}
A \emph{multiset} is a generalization of the notion of set where each element may appear more than once.
The \emph{multiplicity} of the element $x$ in the multiset $M$, denoted by $m_M(x)$, is the number of occurrences of $x$ in $M$.
We adopt the natural convention that $m_M(x)=0$ if $x$ does not belong to $M$.
The cardinality of a multiset $M$ denoted by $|M|$ is its number of elements \emph{with their multiplicity}: $\Sigma_{x}m_M(x)$.
If $M$ and $N$ are two multisets, $M \cup N$ is the multiset $A$ such that $\forall x$, $m_A(x)=m_M(x)+m_N(x)$, and $M \setminus N$ is the multiset $D$ such that $\forall x$, $m_D(x)=\max(0,m_M(x)-m_N(x))$.
We write $M \subseteq N$ if and only if $M \setminus N = \emptyset$ and $M \subset N$ if and only if $M \subseteq N$ and $M \neq N$.

\begin{example}
Let $M=\{1,2,2,4,5,5,5\}$ and $N=\{1,1,1,2,2,3,3,4,5,5,5,5\}$.
Then, $|M|=7$, $|N|=12$, $M \setminus N = \emptyset$, $N \setminus M = \{1,1,3,3,5\}$, and $M \subseteq N$.
\end{example}

\textbf{\graphm.} 
The problem is defined as follows:

\Pb{\graphm}{A triple $(G,c,M)$, where $G=(V,E)$ is a graph, $c: V \rightarrow \mathcal C$ 
 is a coloring of the vertices, and $M$ is a multiset of colors of $\mathcal C$.}{A subset $R \subseteq V$ such that \\(1) $G[R]$ is connected and \\(2) $c(R)=M$. }{Output}
In the above definition, $c(R)$ denotes the multiset of colors of vertices in $R$.
We use that slight abuse of notation for convenience.
We will refer to condition (1) as the \emph{connectivity constraint} and to condition (2) as the \emph{multiset constraint}.

\medskip

\textbf{Parameterized Complexity.}
A parameterized problem $(I,k)$ is said \emph{fixed-parameter tractable} (or in the class $\fpt$) w.r.t. (with respect to) parameter $k$ if it can be solved in $f(k)\cdot|I|^c$ time (in \emph{fpt-time}), where $f$ is any computable function and $c$ is a constant (see \cite{Downey1999,Niedermeier2006,Cygan15} for more details about fixed-parameter tractability).
The parameterized complexity hierarchy is composed of the classes $\fpt \subseteq \wone \subseteq \wtwo \subseteq \dots \subseteq \wpp \subseteq \mathsf{XP}$. 
The class $\xp$ is the set of problems solvable in time $|I|^{f(k)}$, where $f$ is a computable function.

{
A $\wone$-hard problem is not fixed-parameter tractable (unless $\fpt=\wone$) and one can prove $\wone$-hardness by means of a \emph{parameterized reduction} from a $\wone$-hard problem.
This is a mapping of an instance~$(I,k)$ of a problem~$A_1$ in $g(k)\cdot |I|^{O(1)}$ time (for any computable function~$g$) into an instance $(I',k')$ for~$A_2$ such that $(I,k)\in A_1\Leftrightarrow (I',k')\in A_2$ and $k'\le h(k)$ for some function~$h$.
}

A powerful technique to design parameterized algorithms is \emph{kernelization}. 
In short, kernelization is a polynomial-time self-reduction algorithm that takes an instance $(I,k)$ of a parameterized problem $P$ as input and computes an equivalent instance $(I',k')$ of $P$ such that $|I'| \leqslant h(k)$ for some computable function $h$ and $k' \leqslant k$. 
The instance $(I',k')$ is called a \emph{kernel} in this case. 
If the function $h$ is polynomial, we say that $(I',k')$ is a polynomial kernel. 

{
It is well known that a decidable problem is in $\fpt$ if and only if it has a kernel, but this equivalence yields super-polynomial kernels (in general). 
To design efficient parameterized algorithms, a kernel of polynomial (or even linear) size in $k$ is important. 
However, some lower bounds on the size of the kernel can be shown under the assumption that the polynomial hierarchy is a proper hierarchy.
To show such results, we will use the cross-composition technique developed by Bodlaender et al.~\cite{Bodlaender2014}. 
}

{
\begin{definition}[Polynomial equivalence relation \cite{Bodlaender2014}]
An equivalence relation $\R$ on $\Sigma^*$ is said to be \emph{polynomial} if the following two conditions hold: \\
(i) There is an algorithm that given two strings $x, y \in \Sigma^*$ decides whether $x$ and $y$ belong to the same equivalence class in time $(|x| + |y|)^{O(1)}$. \\
(ii) For any finite set $S \subseteq \Sigma^*$ the equivalence relation $\R$ partitions the elements of $S$ into at most $(\max_{x \in S} |x|)^{O(1)}$ classes.
\end{definition}

\begin{definition}[OR-cross-composition \cite{Bodlaender2014}] 
Let $L \subseteq \Sigma^*$ be a set and let $Q \subseteq \Sigma^* \times \mathbb{N}$ be a parameterized problem. 
We say that $L$ \emph{cross-composes} into $Q$ if there is a polynomial equivalence relation $\R$ and an algorithm which, given $t$ strings $x_1, x_2,  \dots, x_t$ belonging to the same equivalence class of $\R$, computes an instance $(x^*,k^*) \in \Sigma^* \times \mathbb{N}$ in time polynomial in $\sum_{i=1}^{t}|x_i|$ such that: \\
(i) $(x^*,k^*) \in Q \Leftrightarrow x_i \in L$ for some $1 \leqslant  i \leqslant t$; and \\
(ii) $k^*$ is bounded by a polynomial in $\max_{i=1}^t |x_i| + \log t$.
\end{definition}

\begin{theorem}[\cite{Bodlaender2014}]
Let $L \subseteq \Sigma^*$ be a set which is $\np$-hard under Karp reductions. If 
$L$ cross-composes into the parameterized problem $Q$, then $Q$ has no polynomial 
kernel unless $\np \subseteq \mathsf{coNP}/poly$.
\end{theorem}
}

\textbf{(Strong) Exponential Time Hypothesis.}
The \emph{Exponential Time Hypothesis} (ETH) is a conjecture by Impagliazzo et al.~\cite{ImpagliazzoETH} asserting that there is no $2^{o(n)}$-time algorithm for \textsc{3-SAT} on instances with $n$ variables.
The so-called sparsification lemma, also proved in \cite{ImpagliazzoETH}, shows that if ETH turns out to be true, then there is no $2^{o(n+m)}$-time algorithm solving \textsc{3-SAT} where $m$ is the number of clauses.
The \emph{Strong Exponential Time Hypothesis} (SETH) by Impagliazzo and Paturi~\cite{Impagliazzo01} further asserts that, for every $\delta < 1$, there is an integer $k$ such that \textsc{$k$-SAT} cannot be solved in time $O(2^{\delta n})$.
Cygan et al. showed that, assuming SETH, for any $\delta < 1$, some problems such as \textsc{Hitting Set} could not be solved in time $O(2^{\delta n})$ either \cite{CyganD12}, where $n$ is the number of elements.
The authors also conjectured that the same result should hold for the \textsc{Set Cover} problem, and gave some supporting pieces of evidence.
We will refer to the assumption that, for any $\delta < 1$, \textsc{Set Cover} instances with $n$ elements cannot be solved in time $O(2^{\delta n})$ as SCH (for \textsc{Set Cover}-hardness).
We insist on the fact that the implication SETH $\Rightarrow$ SCH is not known yet.

\textbf{Previous work.}
Many results about the complexity of \graphm are known. 
The problem is $\np$-hard even with strong restrictions. 
For instance, it remains $\np$-hard for bipartite graphs of maximum degree $4$ and motifs containing two colors only~\cite{fellows2011}, or for trees of maximum degree $3$ and when the motif is colorful (that is, no color occurs more than once)~\cite{fellows2011}, or for rooted trees of depth $2$~\cite{Ambalath2010}. 
However, the problem is solvable in polynomial time when the graph is a caterpillar~\cite{Ambalath2010}, or when both the number of colors in the motif and the treewidth of the graph are bounded by a constant~\cite{fellows2011}.

As \graphm is intractable even for very restricted classes of graphs, and considering that, in practice, the motif is supposed to be small compared to the graph, the parameterized complexity of \graphm relatively to the size of the motif has been tackled. 
It is indeed in $\fpt$ when parameterized by the size of the motif. 
At least seven different papers gave an FPT algorithm \cite{fellows2011,Betzler2011,guillemotfinding,Koutis2012,Bjorklund2016,Pinter14,PinterSZ14}.
The best (randomized) algorithm runs in time $O^*(2^k)$ where the $O^*$ notation suppresses polynomial factors~\cite{Bjorklund2016,Pinter14} and works well in practice for small values of $k$, even with hundreds of millions of edges~\cite{Bjorklund2015}.
The current best deterministic algorithm takes time $O^*(5.22^k)$~\cite{PinterSZ14}. 
However, an algorithm running in time $O^*((2-\epsilon)^k)$ would break the $2^n$ barrier in solving \textsc{Set Cover} instances with $n$ elements (that is, would disprove SCH)~\cite{Bjorklund2016}. 
Besides, it is unlikely that \graphm admits a polynomial kernel, even on a restricted class of trees~\cite{Ambalath2010}.
Ganian proved that \graphm is in $\fpt$ when the parameter is the size of a minimum vertex cover of the graph~\cite{Ganian11}. 
Actually, his algorithm is given for a smaller parameter called twin-cover. 
Ganian also showed that \graphm can be solved in $O^*(2^k)$ for graphs with neighborhood diversity $k$~\cite{Ganian2012}. 
On the negative side, the problem is $\wone$-hard with respect to the number of colors, even for trees \cite{fellows2011}.
To deal with the huge rate of noise in the biological data, many variants of the problem has been introduced. 
For example, the approach of Dondi \emph{et al.} requires a solution with a minimum number of connected components~\cite{Dondi2011a}, while the one of Betzler \emph{et al.} asks for a $2$-connected solution~\cite{Betzler2011}. 
In other variants stemming purely from bio-informatics, some colors can be added to
, substituted or subtracted from the solution \cite{Bruckner2009,Dondi2011a}. 


In light of the previous paragraphs, it is clear that the complexity of \graphm is well known for different versions and constraints on the problem itself. 
However, only few works take into account the structure of the input graph.  
We believe that this an interesting direction since \graphm has applications in real-life problems, where the input is not random. 
For example, some biological networks have been shown scale-free or with small diameter~\cite{Alm2003}.
We will therefore introduce a systematic study with respect to structural graph parameters~\cite{Niedermeier2012,FellowsLMMRS09}. 
We believe that this is also of theoretical interest, to understand how a given parameter influences the complexity of the problem.

 


\textbf{Our contribution.} 
In Section~\ref{sec:easy}, we improve the known FPT algorithms with parameter distance to clique, vertex cover number, and edge clique cover number. 
We also give a parameterized algorithm for the parameter distance to co-cluster 
which nicely reuses the FPT algorithms for both vertex cover number and distance to clique and another algorithm for parameter vertex clique cover number.
These last two algorithms are noteworthy since a bounded distance to co-cluster or a bounded vertex clique cover number do not imply a bounded neighborhood diversity, a parameter for which \graphm was already known to be in $\fpt$. 
We also show that a polynomial kernel for the aforementioned parameters is unlikely and give some ETH-based lower bounds for the worst case running time.
In Section~\ref{sec:hard}, we show that \graphm remains hard on graphs of constant distance to disjoint paths, or constant bandwidth, or constant distance to cluster, or constant dominating set number. 
More surprisingly, we establish that \graphm is $\wone$-hard (but in $\wpp$) for the parameter max leaf number.
To the best of our knowledge, there is no previously known problem behaving similarly when parameterized by max leaf number.
Indeed, graphs with bounded max leaf number are really simple and, for instance, all the problems studied in \cite{FellowsLMMRS09} are FPT for this parameter.
These positive and negative results draw a tight line between tractability and intractability (see Figure~\ref{fig:recap}).


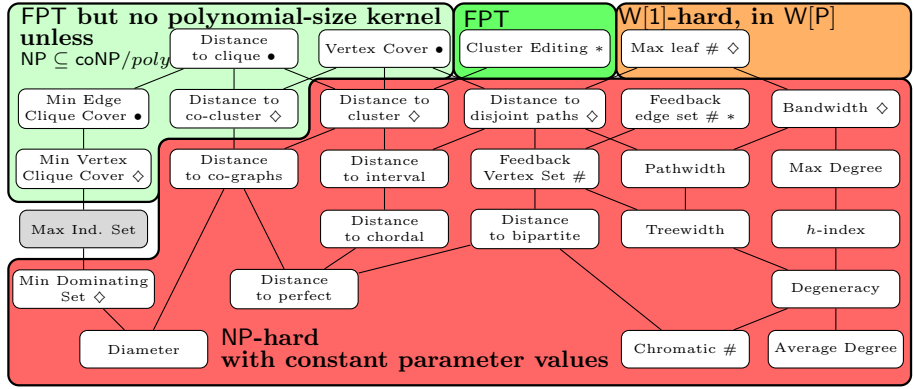
\begin{figure}[]
\begin{center}
{\scalefont{0.5} 
\begin{tikzpicture}[scale=0.8,every text node part/.style={align=center}]

      \path [draw=black, fill=red!60, line width=0.3mm,rounded corners]
            (-5.75,-5.6) -- (-5.75,-3.5) -- (-3.25,-3.5) -- (-3.25,-1.5) -- (-0.75,-1.5) -- (-0.75,-0.5)-- (9.25,-0.5) --  (9.25,-1.25) -- (9.25,-5.6)-- cycle;

\node () at (-1.4,-4.8) [ font=\bfseries] {{\small $\np$-hard}};
\node () at (1.0,-5.2) [ font=\bfseries] {{\small with constant parameter values}};

\path [draw=black, fill=green!20, line width=0.3mm,rounded corners]
	(1.65,0.75) -- (1.65, -0.5)  -- (-0.75, -0.5) -- (-0.75, -1.5) -- (-3.25, -1.5) -- (-3.25, -2.55) -- (-5.75,-2.55)-- (-5.75, 0.75)-- cycle;
\node () at (-2.05,0.5) [font=\bfseries,minimum width=1.7cm,minimum height=0.5cm] {{\small $\fpt$ but no polynomial-size  kernel}};
\node () at (-4.87,0.2) [font=\bfseries,minimum width=1.7cm,minimum height=0.5cm] {{\small  unless }};
\node () at (-4.3,-0.2) [font=\bfseries,minimum width=1.7cm,minimum height=0.5cm] {{\scriptsize  $\np \subseteq \mathsf{coNP}/poly$}};

\path [draw=black, fill=orange!60, line width=0.3mm,rounded corners]
	(9.25,0.75) -- (9.25, -0.5) -- (4.35,-0.5)  -- (4.35,0.75) -- cycle;
\node () at (6.2,0.5) [ font=\bfseries] {{\small $\wone$-hard, in $\wpp$}};

\path [draw=black, fill=green!60, line width=0.3mm,rounded corners]
	(4.35,0.75) -- (4.35, -0.5) -- (1.65,-0.5)  -- (1.65,0.75) -- cycle;

\node () at (2.2,0.5) [ font=\bfseries] {{\small $\fpt$}};



\node (dist to clique) at (-2,0) [draw,color=black, fill=white,rounded corners=0.1cm,minimum width=1.7cm,minimum height=0.5cm] {Distance \\ to clique $\bullet$};


\node (vertex cover) at (0.5,0) [draw,color=black, fill=white,rounded corners=0.1cm,minimum width=1.7cm,minimum height=0.5cm] {Vertex Cover $\bullet$};


\node (cluster editing) at (3,0) [draw,color=black, fill=white,rounded corners=0.1cm,minimum width=1.7cm,minimum height=0.5cm] {Cluster Editing $\ast$};




\node (leaf) at (5.5,0) [draw,color=black, fill=white,rounded corners=0.1cm,minimum width=1.7cm,minimum height=0.5cm] {Max leaf \# $\diamondsuit$};
%

\node (clique cover) at (-4.5,-1) [draw,color=black, fill=white,rounded corners=0.1cm,minimum width=1.7cm,minimum height=0.5cm] { Min Edge\\ Clique Cover $\bullet$};


\node (co-cluster) at (-2,-1) [draw,color=black, fill=white,rounded corners=0.1cm,minimum width=1.7cm,minimum height=0.5cm] { Distance to\\ co-cluster $\diamondsuit$} ;


\node (distance cluster) at (0.5,-1) [draw,color=black, fill=white,rounded corners=0.1cm,minimum width=1.7cm,minimum height=0.5cm] { Distance to\\ cluster $\diamondsuit$};


\node (paths) at (3,-1) [draw,color=black, fill=white,rounded corners=0.1cm,minimum width=1.7cm,minimum height=0.5cm] { Distance to\\ disjoint paths $\diamondsuit$};


\node (fes) at (5.5,-1) [draw,color=black, fill=white,rounded corners=0.1cm,minimum width=1.7cm,minimum height=0.5cm] {Feedback \\ edge set \# $\ast$};

\node (bandwidth) at (8,-1) [draw,color=black, fill=white,rounded corners=0.1cm,minimum width=1.7cm,minimum height=0.5cm] {Bandwidth $\diamondsuit$};

\node (VCC) at (-4.5,-2) [draw,color=black, fill=white,rounded corners=0.1cm,minimum width=1.7cm,minimum height=0.5cm] {Min Vertex\\Clique Cover $\diamondsuit$};

  %
  \node (IS) at (-4.5,-3) [draw,color=black, fill=gray!30,rounded corners=0.1cm,minimum width=1.7cm,minimum height=0.5cm] {Max Ind. Set};
  
  
  \node (cographs) at (-2,-2) [draw,color=black, fill=white,rounded corners=0.1cm,minimum width=1.7cm,minimum height=0.5cm] {Distance\\ to co-graphs};
  
  
    \node (interval) at (0.5,-2) [draw,color=black, fill=white,rounded corners=0.1cm,minimum width=1.7cm,minimum height=0.5cm] {Distance\\ to interval};
  

    \node (fvs) at (3,-2) [draw,color=black, fill=white,rounded corners=0.1cm,minimum width=1.7cm,minimum height=0.5cm] {Feedback\\Vertex Set \#};


    \node (pathwidth) at (5.5,-2) [draw,color=black, fill=white,rounded corners=0.1cm,minimum width=1.7cm,minimum height=0.5cm] {Pathwidth};

 
\node (degree) at (8,-2) [draw,color=black, fill=white,rounded corners=0.1cm,minimum width=1.7cm,minimum height=0.5cm] {Max Degree};
  %
  
\node (ds) at (-4.5,-4) [draw,color=black, fill=white,rounded corners=0.1cm,minimum width=1.7cm,minimum height=0.5cm] {Min Dominating\\ Set $\diamondsuit$};
  

\node (chordal) at (0.5,-3) [draw,color=black, fill=white,rounded corners=0.1cm,minimum width=1.7cm,minimum height=0.5cm] {Distance \\to chordal};


\node (bipartite) at (3,-3) [draw,color=black, fill=white,rounded corners=0.1cm,minimum width=1.7cm,minimum height=0.5cm] {Distance \\to bipartite};


\node (treewidth) at (5.5,-3) [draw,color=black, fill=white,rounded corners=0.1cm,minimum width=1.7cm,minimum height=0.5cm] {Treewidth};

\node (hindex) at (8,-3) [draw,color=black, fill=white,rounded corners=0.1cm,minimum width=1.7cm,minimum height=0.5cm] {$h$-index};
  %
\node (diameter) at (-3.5,-5) [draw,color=black, fill=white,rounded corners=0.1cm,minimum width=1.7cm,minimum height=0.5cm] {Diameter};


\node (perfect) at (-1,-4) [draw,color=black, fill=white,rounded corners=0.1cm,minimum width=1.7cm,minimum height=0.5cm] {Distance\\to perfect};


\node (degeneracy) at (8,-4) [draw,color=black, fill=white,rounded corners=0.1cm,minimum width=1.7cm,minimum height=0.5cm] {Degeneracy};
  %
\node (chromatic) at (5.5,-5) [draw,color=black, fill=white,rounded corners=0.1cm,minimum width=1.7cm,minimum height=0.5cm] {Chromatic \#};

\node (av degree) at (8,-5) [draw,color=black, fill=white,rounded corners=0.1cm,minimum width=1.7cm,minimum height=0.5cm] {Average Degree};

    \path[draw] (dist to clique) -- (clique cover) -- (VCC) -- (IS) -- (ds) -- (diameter) -- (cographs) -- (co-cluster) -- (dist to clique) -- (distance cluster) -- (cographs) -- (perfect) -- (chordal) -- (interval) -- (distance cluster) -- (vertex cover) -- (co-cluster);
    
     \path[draw]  (vertex cover) -- (paths) -- (interval);
     
     \path[draw] (perfect) -- (bipartite) -- (fvs) -- (paths) -- (leaf) -- (bandwidth) -- (degree) -- (hindex) -- (degeneracy) -- (chromatic) -- (bipartite);
     
     \path[draw] (cluster editing) -- (distance cluster);
     
     \path[draw] (fes) -- (fvs) -- (treewidth) -- (degeneracy) -- (av degree);
     
     \path[draw] (paths) -- (pathwidth) -- (bandwidth);
     
      \path[draw] (treewidth) -- (pathwidth);
      
  
\end{tikzpicture}
} 
\caption{ 
Hasse diagram of the relationship between different parameters (\cite{Niedermeier2012}). 
Two parameters are connected by a line if the parameter below can be polynomially upper-bounded in the parameter above. 
For example, \emph{vertex cover} is above \emph{distance to disjoint paths} since deleting a vertex cover produces an independent set, hence a set of disjoint paths. 
Therefore, positive results propagate upwards, while negative results propagate downwards. 
Results marked by~$\diamondsuit$ are obtained in this paper, those marked with~$\bullet$ are improvement of existing results, and those marked with~$\ast$ are corollaries of existing results.
Parameter neighborhood diversity is not depicted since its relations with vertex cover may be exponential.
We refer to~\cite[Figure 1]{Lampis2012} for a diagram with neighborhood diversity.
We note that neighborhood diversity would be below vertex cover, not comparable to feedback vertex set, patwidth or treewidth, but above cliquewidth (this last would be below treewidth).
}\label{fig:recap}
\end{center}
\end{figure}

\section{$\fpt$ algorithms, kernelization and ETH-based lower bounds}\label{sec:easy}

In this section, we improve or establish new FPT algorithms for several parameters. 
We complement those algorithms with some lower bounds under ETH, SETH, and SCH.
We also give a lower bound on the size of the kernel for all those parameters except \emph{cluster editing number}. 
Figure~\ref{fig:recap} summarizes those results.

\subsection{Cluster editing and linear neighborhood diversity}

The cluster editing number of a graph is the number of edge deletions or additions required to get a cluster graph. 
It can be computed in time $O^*(1.62^k)$~\cite{Bocker2012}. 
We will use a known result involving another parameter called neighborhood diversity introduced by Lampis~\cite{Lampis2012}. 
A graph has neighborhood diversity $k$ if there is a partition of its vertices into at most $k$ sets such that all the vertices in each set \emph{have the same type}.
And, two vertices $u$ and $v$ \emph{have the same type} if $N(v) \setminus \{u\} = N(u) \setminus \{v\}$.
We say that a graph parameter $\kappa$ has \emph{linear} (resp.~\emph{exponential}) \emph{neighborhood diversity} if, for every positive integer $k$, all the graphs $G$ such that $\kappa(G) \leqslant k$ have neighborhood diversity $O(k)$ (resp. $2^{O(k)}$).
We say that a parameter $\kappa$ has \emph{unbounded neighborhood diversity}, if there is \emph{no} function $f$ such that all graphs $G$ with $\kappa(G) \leqslant k$ have neighborhood diversity $f(k)$.

\begin{theorem}[\cite{Ganian2012}]\label{thm:Ganian}
\graphm can be solved in $O^*(2^k)$ on graphs with neighborhood diversity $k$.
\end{theorem}

The following result is a direct consequence of the fact that, restricted to connected graphs, cluster editing has linear neighborhood diversity.
\begin{corollary}
\graphm can be solved in $O^*(8^k)$, where $k$ is the cluster editing number.
\end{corollary}
\begin{proof}
Let $(G=(V,E),c,M)$ be any instance of \graphm.
We can assume that $G$ is connected, otherwise we run the algorithm in each connected component of $G$.
Let $X$ be the set of vertices which are an endpoint of an edited edge (deleted or added) and let $G'$ be the cluster graph obtained by the $k$ edge editions.
We may observe that $|X| \leqslant 2k$ and that the number of maximal cliques $C_1,\ldots,C_l$ in $G'$ is bounded by $k$ (otherwise, $G$ could not be connected).
For each $i \in [l]$, and for each vertex $v \in C_i \setminus X$, $N[v]=C_i$.
Thus the neighborhood diversity of $G$ is bounded by $|X|+l \leqslant 2k+k=3k$.
So, we can run the algorithm for bounded neighborhood diversity \cite{Ganian2012} and it takes time $O^*(2^{3k})$.
\end{proof}


\subsection{Parameters with exponential neighborhood diversity}

The next three parameters that we consider are \emph{distance to clique}, \emph{size of a minimum vertex cover}, and \emph{size of a minimum edge clique cover}.
For the first two, a value of $k$ entails that the neighborhood diversity is at most $k+2^k$; whereas, edge clique cover number $k$ implies that the neighborhood diversity is at most $2^k$.
Therefore, Ganian has already given an algorithm running in double exponential time for these parameters ($O^*(2^{k+2^k})$ or $O^*(2^{2^k})$, see Theorem~\ref{thm:Ganian}, \cite{Ganian11,Ganian2012}).
We improve this bound to single exponential time $2^{O(k)}$ (more precisely $O^*(3^k)$) for distance to clique and to $2^{O(k \log k)}$ for the vertex cover and edge clique cover numbers.
The latter running time is sometimes called \emph{slightly superexponential} FPT time \cite{Lokshtanov11}.
Then, we prove that for each of those three parameters, a polynomial kernel is unlikely.

As a preparatory lemma for the algorithm parameterized by distance to clique, we show that a variant of \textsc{Set Cover} with thresholds is solvable in time $O^*(2^n)$, where $n$ is the size of the universe.
In the problem that we call here \textsc{Colored Set Cover with Thresholds}, one is given a triple $(\U,\S=\C_1 \uplus \ldots \uplus \C_l,(a_1,\ldots,a_l))$ where $\U$ is a ground set of $n$ elements, $\S$ is a set of subsets of $\U$ partitioned into $l$ classes called \emph{colors} and $(a_1,\ldots,a_l)$ is a tuple of $l$ positive integers called \emph{threshold vector}.
The goal is to find a set cover $\T \subseteq \S$ (not necessarily minimum) such that for each $i \in [l]$, the number of sets with color $i$ (that is, in $\C_i$) in $\T$ is at most $a_i$.

\begin{lemma}
\textsc{Colored Set Cover with Thresholds} with $n$ elements and $m$ sets can be solved in time $O(nm2^n+nm)$.
\end{lemma}

\begin{proof}
We order the sets of $\S$ such that sets of the same color appear consecutively, say, first the sets of $\C_1$, then the sets of $\C_2$, and so on. 
The order within the sets of a same color is not important and is chosen arbitrarily.
We denote the sets resultantly ordered by $S_1, \ldots, S_m$ and function $c$ maps the index of a set to its color.
Therefore, $c(j)=i$ means that set $S_j$ has color $i$ ($S_j \in C_i$).
We fill by dynamic programming the table $T$, where $T[U,j]$ is meant to contain the minimum number of sets in $\C_{c(j)}$ among any subset of $\{S_1,\dots,S_j\}$ that covers $U \subseteq \U$ and respects the threshold vector.

As an initialization step, for each $U \subseteq \U$, we set $T[U,1] = 1$ if $U \subseteq S_1$, and $T[U,1] = \infty$ otherwise.
For each $j \in [2,m]$, assuming that $T[U',j-1]$ was already filled for every $U' \subseteq \U$, we distinguish two cases to fill $T[U,j]$.
If $S_j$ is the first set of the color class $\C_{c(j)}$ then:

$T[U,j] = 
\left \{
         \begin{array}{l}
               0 \text{ if } T[U,j-1]<\infty \text{~~~~~~~~~~~~~~~~~~~~~~~~~~~~~~~~~~~~(* discard } S_j \text{ *)}\\
               
               1 \text{ if } T[U,j-1]=\infty \text{ and } T[U \setminus S_j, j-1] < \infty \text{~ (* add } S_j \text{ *)}\\

               \infty \text{ otherwise}
            \end{array}
         \right.
 $
 
\medskip
Otherwise $S_j$ is not the first set in $\C_{c(j)}$ and:

$T[U,j] = \min
\left \{
	\begin{array}{l}
		T[U,j-1] \text{~~~~~~~~~~~~~~~~~~~~~~~~~~~~~~~~~~~~~~~~~~ (* discard } S_j \text{ *)}\\
		v+1 \text{ if } v<a_{c(j)} \text{ and } \infty \text{ otherwise } \text{~~~~~~~~~~ (* add } S_j \text{ *)}
	\end{array}
\right.
$

with $v=T[U \setminus S_j,j-1]$.

\medskip
A standard induction shows that the instance is positive if and only if $T[\U,m] \neq \infty$.
The only costly operation in filling one entry of table $T$ is the set difference which can be done in $O(n)$ time.
If we want to produce an actual solution (and not solely decide the problem), we can add one bit in each entry $T[U,j]$ signaling whether or not $S_j$ should be taken.
Should the instance be positive, it then takes time $O(nm)$ to reconstruct a solution from a filled table $T$.
Therefore, the running time is $O(n|T|+nm)=O(nm2^n+nm)$.
\end{proof}

\begin{theorem}\label{th:distclique}
\graphm can be solved in $O^*(3^k)$, where $k$ is the distance to clique.
\end{theorem}

\begin{proof}
Let $(G=(V,E),c:V \rightarrow \mathcal C,M)$ be any instance of \graphm and assume $R$ is a solution, that is $G[R]$ is connected and $c(R)=M$.
If there is no solution, our algorithm will detect this eventually.
We first compute a set $S \subseteq V$ of size $k$ such that $C:=V \setminus S$ is a clique.
This can be done in time $O^*(2^k)$ by branching over the two endpoints of a \emph{non-edge}, or even in time $O^*(1.2738^k)$ by applying the state-of-the-art algorithm for \textsc{Vertex Cover} on the complementary graph \cite{Chen10}.  
Running through all the $2^k$ subsets of $S$, one can guess the subset $S'= R \cap S$ of $S$ which is in the solution $R$.
Let $S_1, S_2, \ldots, S_{k'}$ be the $k' \leqslant k$ connected components of $G[S']$.
It must hold that $c(S') \subseteq M$, otherwise $R$ would not be a solution.
Now, the problem boils down to finding a non-empty (an empty subset would mean that $S'=R$ which can be easily checked) subset $C' \subseteq C$ such that $G[S' \cup C']$ is connected and $c(C') \subseteq M \setminus c(S')$.
Then, the set $S' \cup C'$ can be extended into a solution by adding vertices of $C \setminus C'$ with the right colors.
The graph $G[S' \cup C']$ is connected if and only if each connected component $S_j$ of $G[S']$ has at least one neighbor in $N(C')$.
We build an equivalent instance of \textsc{Colored Set Cover with Thresholds} in the following way.
The ground set $\U$ is of size $k'$ with one element $x_j$ per connected component $S_j$ of $G[S']$.
For each vertex $v$ in $C$ colored by $i$, there is a set $S_v$ colored by $i$ such that $x_j \in S_v$ if and only if $N(v) \cap S_j \neq \emptyset$.
For each color $i$, the threshold $a_i$ is set to the multiplicity of $i$ in $M \setminus c(S')$.
The number of elements is $k'$ and the number of sets is polynomial. 
So, it takes time $O^*(2^{k'})$ to solve this instance.
Therefore, the overall running time is $O^*(2^k+\sum\limits_{S': S'\subseteq S} 2^{|S'|})=O^*(2^k+\sum\limits_{0 \leqslant k' \leqslant k} {k \choose k'} 2^{k'})=O^*(3^k)$. 
\end{proof}





\begin{theorem}\label{th:vcfpt}
\graphm can be solved in $O^*(2^{2k \log k})$ on graphs with a vertex cover of size $k$.
\end{theorem}

\begin{proof}

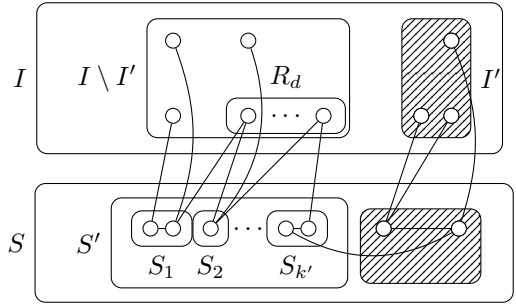
\begin{figure}[ht!]
\centering
\begin{tikzpicture}[scale=0.9,auto]

\node[vertex] (S11) {};
\node[vertex, right of=S11,node distance=0.3cm] (S12)  {};
\node[draw, rectangle,rounded corners,label=below:$S_1$,fit= (S11) (S12)] (S1) {};
\draw (S11) -- (S12) ;

\node[vertex,right of=S12,node distance=0.5cm] (S21) {};
\node[draw, rectangle,rounded corners,label=below:$S_2$,fit=(S21)] (S2) {};

\node[below of=S2,node distance=0.5cm] (bS2) {};

\node[right of=S21,node distance=0.53cm] (S2Sk) {\dots};

\node[vertex,right of=S21] (Sk1) {};
\node[vertex,right of=Sk1,node distance=0.3cm] (Sk2) {};
\node[draw, rectangle,rounded corners,label=below:$S_{k'}$,fit=(Sk1) (Sk2)] (Skp) {};
\draw (Sk1) -- (Sk2) ;

\node[draw,rectangle,rounded corners,fit= (S1) (bS2) (Skp),label={[name=lSp]left:$S'$}] (Sp) {};

\node[vertex,right of=Sk2] (Sr1) {};
\node[vertex,right of=Sr1] (Sr2) {};
\node[below of=Sr1,node distance=0.45cm] (Sr3) {};
\node[draw,rectangle,rounded corners,fit=(Sr1) (Sr2) (Sr3),pattern=north east lines] (Sremove) {};
\node at (Sr1) [vertex,fill=white] {};
\node at (Sr2) [vertex,fill=white] {};

\draw (Sr1) -- (Sr2) edge[bend left=30] (Sk1) ;

\node[draw, rectangle,rounded corners,label=left:$S$,fit=(lSp) (Sp) (Sremove)] (S) {};

\node[vertex,above of=S12,node distance=1.5cm] (I1) {};
\node[vertex,right of=I1] (I2) {};
\node[vertex,right of=I2] (I3) {};
\node[vertex,above of=I1] (I4) {};
\node[vertex,right of=I4] (I5) {};

\node[right of=I2,node distance=0.53cm] (dI23) {\dots};

\node[draw, rectangle, rounded corners,fit=(I2) (I3),label=$R_d$] {};

\draw (I1) -- (S11);
\draw (I2) -- (S12);
\draw (I2) -- (S21);
\draw (I3) -- (S21);
\draw (I3) -- (Sk2);
\draw (I4) edge[bend left=20] (S12);
\draw (I5) edge[bend left] (S21);
\node[draw,rectangle,rounded corners,label={[name=IminI]left:$I \setminus I'$},fit=(I1) (I2) (I3) (I4) (I5)] (If) {};

\node[vertex,right of=I3,node distance=1.3cm] (I6) {};
\node[vertex,right of=I6, node distance=0.4cm] (I7) {};
\node[vertex,above of=I7] (I8) {};
\draw (I6) -- (Sr1);
\draw (I7) -- (Sr1);
\draw (I8) edge[bend left=20] (Sr2);
\node[draw,rectangle,rounded corners,label=right:$I'$,fit=(I6) (I7) (I8),pattern=north east lines] (Iremove) {};

\node at (I6) [vertex,fill=white] {};
\node at (I7) [vertex,fill=white] {};
\node at (I8) [vertex,fill=white] {};

\node[draw, rectangle,rounded corners,label=left:$I$,fit=(IminI) (If) (Iremove)] (I) {};

\end{tikzpicture}
\caption{The subsets of $V$ relevant to the algorithm of Theorem~\ref{th:vcfpt}.
}\label{fig:vc}
\end{figure}

We start similarly to the previous algorithm.
We compute a minimum vertex cover $S$ of $G$ in time $O^*(2^k)$ (or $O^*(1.2738^k)$ \cite{Chen10}), and then guess in time $O^*(2^k)$ the subset $S'=S \cap R$, where $R$ is a fixed solution.
Again, we denote by $S_1, S_2, \ldots, S_{k'}$ the connected components of $G[S']$.
We remove $c(S')$ from the motif and we remove from $V$ the set $I'$ of the vertices of the independent set $I:=V \setminus S$ which have no neighbor in $S'$ (see Figure~\ref{fig:vc}).
Now, by the transformation presented in the algorithm parameterized by distance to clique, the problem could be made equivalent to a constrained version of \textsc{Colored Set Cover with Thresholds} where the intersection graph (with an edge between two sets if they have a non-empty intersection) of the solution has to be connected.  
Unfortunately, it is not clear whether or not this variant can be solved in time $2^{O(n)}$.
Thus, at this point, we have to do something different.

Let $R_d=\{r_1,r_2,\ldots,r_l\} \subseteq R \setminus S'$ be a minimal (inclusion-wise) set of vertices such that $G[S' \cup R_d]$ is connected.
We can observe that $l \leqslant k' \leqslant k$.
We guess in time $O^*(l!B_l)$ (where $B_l$ is the $l$-th Bell number, i.e., the number of partitions of a set of size $l$) an ordered partition $P:=\langle A_1,A_2,\ldots,A_l \rangle$ of the connected components $\{S_1,\ldots,S_{k'}\}$ such that, for each $i \in [l]$, (1) $r_i$ has at least one neighbor in each connected component of $A_i$ and (2) if $i \geqslant 2$, $r_i$ has at least one neighbor in a connected component of $\bigcup_{1 \leqslant j < i} A_j$. 
Note that such an ordered partition always exists since $G[S' \cup R_d]$ is connected.
Now, we build the bipartite graph $B=(P \cup M',F)$, where $M'=M \setminus c(S')$ and there is an edge between $A_i \in P$ and each copy of color $c \in M'$ if and only if there is a vertex $v \in I$ colored by $c$ in the original graph $G$ and such that (1) $v$ has at least one neighbor in each connected component of $A_i$ and (2) if $i \geqslant 2$, $v$ has at least one neighbor in a connected component of $\bigcup_{1 \leqslant j < i} A_j$.
By construction, $\{\{A_i,c(r_i)\}$ $|$ $i \in [l]\}$ is a maximum matching of size $|P|=l$ in graph $B$.
Thus, we compute in polynomial time a maximum matching $\{\{A_i,c_i\}$ $|$ $i \in [l]\}$ in $B$. 
Then, we obtain a solution to the \graphm instance by taking, for each $i \in |l]$ any vertex $v_i$ colored by $c_i$ and having (1) at least one neighbor in each connected component of $A_i$ and (2) if $i \geqslant 2$, at least one neighbor in a connected component of $\bigcup_{1 \leqslant j < i} A_j$.
This can also be done in polynomial time and the existence of such a $v_i$ is guaranteed by the construction of graph $B$.
Then, we complete set $S' \cup \bigcup_{i \in [l]}\{v_i\}$ into a solution by taking any vertices in $I \setminus I'$ with the right colors.
As $l! \leqslant l^l$, $B_l \leqslant (\frac{l}{2})^l$ (even $B_l < (\frac{0.792 l}{\ln{(l+1)}})^l$ \cite{Berend10}), and $l \leqslant k$ the overall running time is $O^*(2^k+2^kk!B_k)=O^*(k^kk^k)=O^*(2^{2k \log k})$. 
\end{proof}

In the \textsc{Edge Clique Cover} problem, one asks, given a graph $G=(V,E)$ and an integer $k$, for $k$ subsets $C_1,\ldots,C_k \subseteq V$, such that $\forall i \in [k]$, $G[C_i]$ is a clique, and $\forall e \in E$, $e$ lies in a clique $C_i$ for some $i \in [k]$.
The set $\{C_1,\ldots,C_k\}$ is called an \emph{edge clique cover} of $G$.
The \emph{edge clique cover number} of a graph $G$ is the smallest $k$ such that $G$ has an edge clique cover of size $k$.
\textsc{Edge Clique Cover} admits a kernel of size $2^k$ (which can be obtained in $O(n^4)$ time)~\cite{Gramm08} and, as observed in \cite{Cygan13}, it can be solved by dynamic programming in time $2^{O(n+m)}$.
Therefore, it can be solved in time $2^{O(2^k+2^{2k})} + O(n^4)$, that is $2^{2^{O(k)}} + O(n^4)$. 
On the negative side, \textsc{Edge Clique Cover} cannot be solved in time $2^{2^{o(k)}}$ under ETH~\cite{Cygan13}.
But, we may imagine that the instance comes with an optimal or close to optimal edge clique cover, or that we have a good heuristic to compute it (a polynomial time approximation with sufficiently good ratio is unlikely \cite{Lund94}).

\begin{theorem}\label{th:edge-clique-cover}
\graphm can be solved in time $2^{2^{O(k)}}  + O(n^4)$, where $k$ is the edge clique cover number, and in time $O^*(2^{2k \log k + k})$ if an edge clique cover of size $k$ is given as part of the input.
\end{theorem}

{
\begin{proof}
Let $I=(G=(V,E),c,M)$ be any instance of \graphm.
If not given, we first compute an edge clique cover $\{C_1,\ldots,C_k\}$ of size $k$ in $G$, in time $2^{2^{O(k)}} + O(n^4)$~\cite{Gramm08,Cygan13}.

We guess in time $O^*(2^k)$ the exact subset $\{C'_1,\ldots,C'_{k'}\} \subseteq \{C_1,\ldots,C_k\}$ of cliques $C_i$ such that $C_i \cap R$ is non-empty, for a fixed solution $R$.
Now, we turn the instance into an equivalent instance where the motif has size $|M|+k'$ and the graph has at most $|V|+k'$ vertices and a vertex cover of size $k'$.
The new graph is a bipartite graph $B=(A \cup W,F)$ such that $A$ contains one vertex $v(C'_i)$ per clique $C'_i$ (so, $A$ is a vertex cover of graph $B$ of size $k' \leqslant k$), $W=C'_1 \cup \ldots \cup C'_{k'} \subseteq V$, and there is an edge in $F$ between $v(C'_i) \in A$ and $w \in W$ if and only if $w \in C'_i$.
Each vertex in $W$ keeps the color it had in $G$.
A fresh color $\gamma$ is given to the $k'$ vertices of $A$, and color $\gamma$ is added to the motif $M$ with multiplicity $k'$.
This coloring is denoted by $c'$ and $M' := M \cup \{\gamma, \ldots, \gamma~(k'~\text{times})\}$.
We run on the instance $I'=(B,c',M')$ the algorithm parameterized by the vertex cover number of Theorem~\ref{th:vcfpt}. 
This algorithm has an overall running time of $O^*(2^k2^{2k \log k})$, if the edge clique cover is given, and $2^{2^{O(k)}} + O(n^4)$ otherwise.

We now explain why the reduction is correct.
We first claim that the set $A \cup R$ is a solution for the instance $I'$.
The colors of $A \cup R$ consist of $k'$ occurences of $\gamma$ plus the colors of $M$ which matches the multiset $M'$. 
Now, we show that $B[A \cup R]$ is connected by reporting a path from any pair $x, y$ of vertices in $A \cup R$.
Let $\psi: A \cup R \rightarrow R$ be the identity function when restricted to $R$ and map vertex $v(C'_i) \in A$ to an arbitrary fixed vertex of $C'_i \cap R$.
By construction $C'_i \cap R \neq \emptyset$, so $\psi$ is well-defined.
As $G[R]$ is connected there is a path between $\psi(x)$ and $\psi(y)$ in $G[R]$: $\psi(x)=u_1, u_2, \ldots, u_h=\psi(y)$.
By definition of a clique cover, any two consecutive vertices $u_\ell$ and $u_{\ell+1}$ ($\ell \in [h-1]$) along this path are in a same clique $C'_i$.
Therefore, in $B[A \cup R]$ there is a corresponding path $u_\ell, v(C'_i), u_{\ell+1}$.
Also $\psi(x)$ (resp.~$\psi(y)$) is either $x$ (resp.~$y$) or linked by an edge to $x$ (resp.~$y$).
Overall, this gives a path from $x$ to $y$ in $B[A \cup R]$.  

Conversely, assume there is a solution $S$ to $I'$.
Set $S$ has to contain $A$ otherwise the color $\gamma$ is not represented $k'$ times.
So, $S=A \uplus R'$.
We claim that $R'$ is a solution for the instance $I$.
In order to match the colors of $M'$, the colors of $R'$ should match the multiset constraint of $M$.
It remains to argue why $G[R']$ is connected.
Let $x,y$ be any two vertices of $R'$.
Since $B$ is bipartite and $B[S]$ is connected, there is a path in $B[S]$: $x=u_1,v(C'_{i_1}),u_2,v(C'_{i_2}),u_3,\ldots,u_{h-1},v(C'_{i_{h-1}}),u_h=y$ with $u_\ell \in R'$ for any $\ell \in [h]$.
As $u_\ell$ and $u_{\ell+1}$ are in the same clique $v(C'_{i_\ell})$ they are linked by an edge in $G[R']$.
Thus, $x=u_1, u_2, u_3, \ldots, u_h=y$ is a path in $G[R']$.
\end{proof}
}



The correctness of the reduction crucially relied on the fact that every edge is fully contained in at least one clique of the cover.
This would not be the case with a vertex clique cover (a partition of the vertex set into sets inducing cliques).
In Section~\ref{subsec:unbounded-nd}, we give a more complicated FPT algorithm parameterized by the vertex clique cover size (if such a cover is given in the input).
It is not surprising that the edges going from one clique to another play an important role in the greater difficulty of the parameterized algorithm.

Ganian~\cite{Ganian11}, Theorem~\ref{th:vcfpt} and Theorem~\ref{th:distclique} prove that \graphm is in $\fpt$ if the parameter is the vertex cover number or the distance to clique.
Therefore, the problem has a kernel for these two parameters~\cite{Niedermeier2006}.
Though, this does not imply that the size of the corresponding kernels is polynomial.
We show that the corresponding kernels cannot be polynomial unless  $\np \subseteq \mathsf{coNP}/poly$. 


\begin{theorem}\label{th:nokernelvc}
Unless  $\np \subseteq \mathsf{coNP}/poly$, \graphm has no polynomial kernel when parameterized by the vertex cover number or the distance to clique, even for (i) motifs with only 3 colors or (ii) when the motif is colorful.
\end{theorem}
\begin{proof}
We only detail the proof for (i) for parameter vertex cover.
\sloppy We will define an OR-cross-composition~\cite{Bodlaender2014} from the $\np$-complete \pxtc problem, stated as follows: given an integer $q$, a set $X = \{x_1,x_2,\dots,x_{3q}\}$ and a collection $\S = \{S_1,\dots,S_{|\S|}\}$ of 3-elements subsets of $X$, the goal is to decide if $\S$ contains a subcollection $\T \subseteq \S$ such that $|\T|=q$ and each element of $X$ occurs in exactly one element of $\T$. Given $t$ instances, $(X_1,\S_1), (X_2,\S_2), \dots, (X_t,\S_t)$, of \pxtc, we define our equivalence relation $\R$ such that any strings that are not encoding valid instances are equivalent, and $(X_i,\S_i), (X_j,\S_j)$ are equivalent if and only if $|X_i| = |X_j|$ and $|\S_i| = |\S_j|$. 
We will build an instance $(G,c,M)$ of \graphm parameterized by the vertex cover number, 
where $G$ is the input graph, $c$ the coloring function and $M$ the motif, such that there is a solution for \graphm if and only if there is an $i \in [t]$ such that there is a solution for $(X_i,\S_i)$. We will now describe how to build such instance of \graphm. The graph $G$ consists of 
$t$ independent nodes $r_1,r_2,\cdots, r_t$.
There are also $O((3q)^3)$ nodes $s_{x,y,z}, 1 \leq x<y<z \leq 3q$, with an edge between $r_i$ and $s_{x,y,z}$ if and only if the 3-element subset $\{x,y,z\}$ exists in $\S_i$. 
Finally, there are $|X_i|=3q$ nodes $x_i, 1 \leq i \leq 3q$, and there is an edge between $x_i$ and every subset $s_{x,y,z}$ where $x_i$ occurs (see Figure~\ref{fig:wratigan}). The coloration is $c(r_i) = 1$, for all $1 \leq i \leq t$, $c(s_{x,y,z}) = 2$ for all $1 \leq x<y<z \leq 3q$, and $c(x_i) = 3, 1 \leq i \leq 3q$. The multiset $M$ consists of 1 occurrence of the color 1, $q$ occurrences of color 2 and $3q$ occurrences of color 3.

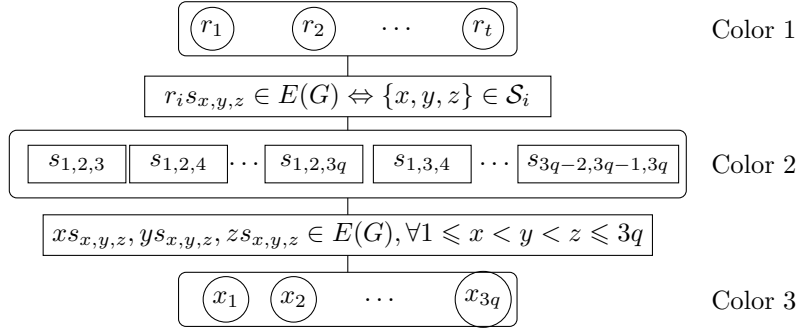
\begin{figure}[ht!]
\centering
\begin{tikzpicture}[scale=0.9,auto]

\node () at (5,-1) {Color 1};
\node () at (5,-3) {Color 2};
\node () at (5,-5) {Color 3};


\draw[color=black,fill=white,rounded corners=0.1cm] (-3.5,-1.4) rectangle (1.5,-0.6);
\node[vertex] (t1) at (-3,-1) {$r_1$};
\node[vertex] (t2) at (-1.5,-1) {$r_2$};
\node () at (-0.25,-1) {$\dots$};
\node[vertex] (tt) at (1,-1) {$r_t$};

\draw[color=black,fill=white,rounded corners=0.1cm] (-6.0,-3.5) rectangle (4.0,-2.5);
\node [draw,minimum width=1.3cm,minimum height=0.5cm] (e12) at (-5.0,-3) {$s_{1,2,3}$};
\node [draw,minimum width=1.3cm,minimum height=0.5cm] (e13) at (-3.5,-3) {$s_{1,2,4}$};
\node () at (-2.5,-3) {$\dots$};
\node[draw,minimum width=1.3cm,minimum height=0.5cm] (e1n) at (-1.5,-3) {$s_{1,2,3q}$};
\node[draw,minimum width=1.3cm,minimum height=0.5cm] (e23) at (0.1,-3) {$s_{1,3,4}$};
\node () at (1.2,-3) {$\dots$};
\node [draw,minimum width=2.0cm,minimum height=0.5cm] (en-1n) at (2.7,-3) {$s_{3q-2,3q-1,3q}$};

\draw[color=black,fill=white,rounded corners=0.1cm] (-3.5,-5.4) rectangle (1.5,-4.6);
\node[vertex] (v11) at (-2.8,-5) {$x_1$};
\node[vertex] (v12) at (-1.8,-5) {$x_2$};
\node () at (-0.5,-5) {$\dots$};
\node[vertex] (v1n) at (1,-5) {$x_{3q}$};


\draw[color=black,fill=white,rounded corners=0.0cm] (-4,-2.3) rectangle (2,-1.7);
\node () at (-1,-2) {$r_is_{x,y,z} \in E(G) \Leftrightarrow \{x,y,z\} \in \S_i$};

\draw (-1,-1.4) -- (-1,-1.7);
\draw (-1,-2.3) -- (-1,-2.5) ;

\draw[color=black,fill=white,rounded corners=0.0cm] (-5.5,-3.75) rectangle (3.5,-4.35);
\node () at (-1,-4.05) {$xs_{x,y,z}, ys_{x,y,z}, zs_{x,y,z} \in E(G), \forall 1 \leqslant x<y<z \leqslant 3q$};

\draw (-1,-3.5) -- (-1,-3.75);
\draw (-1,-4.35) -- (-1,-4.6);


\end{tikzpicture}
\caption{Illustration of the construction of $G$ for parameter vertex cover. The motif consists of 1 occurrence of color 1, $q$ of color 2 and $3q$ of color 3.}\label{fig:wratigan}
\end{figure}

It is easy to see that 
$\{s_{x,y,z} | 1 \leqslant x < y < z \leqslant 3q\}
\cup \{x_i | 1 \leq i \leq 3q \}
$
is a vertex cover for $G$ (as its removal leaves an independent set) 
and that its size is polynomial in $3q$ and hence in the size of the largest instance.

Let us show that there is a solution for our instance of \graphm if and only if at least one of the $(X_i,\S_i)$'s has a solution of size $q$.

Suppose that $(X_i,\S_i)$ has a solution $\T_i$ of size $q$. 
We set $R = \{r_i\} \cup \{s_{x,y,z} $ $|$ $ \{x,y,z \} \in \T_i\} \cup \{x_i | 1 \leq i \leq 3q\}$. 
One can easily check that $G[R]$ is connected and that $c(R) = M$.

Conversely, suppose now that there is a solution $R \subseteq V$ such that $G[R]$ is connected and $c(R) = M$. 
Due to the motif, only one of the nodes $r_i$ is in $R$ and all nodes $x_i$ are in $R$. 
We claim that there is then a solution $\T_i$ in $(X_i,\S_i)$, where $i$ is the index of the only node $r_i$ in $R$. 
We add in $\T_i$ the $q$ sets $\{x,y,z\}$ such that $s_{x,y,z} \in R$. 
Since $R$ is a solution, the nodes $s_{x,y,z}$ in $R$ correspond to a partition of $X$; otherwise, one of the nodes $x_i$ would be disconnected.
Then, $\T_i$ covers exactly all the elements of $X_i$. 
By the connectivity constraint, the $q$ sets added in $\T_i$ all occur in the instance $i$ such that $r_i \in R$. 

If the considered parameter is the distance to clique, one can consider the nodes $r_1,r_2,\dots,r_t$ as a clique. 
The removal of $\{s_{x,y,z} | 1 \leqslant x < y < z \leqslant 3q\} \cup \{x_i | 1 \leq i \leq 3q \}$ leaves one clique and its size is polynomial in the size of the largest instance.
The correctness is the same as for parameter vertex cover number, as only one occurrence of color 1 is in the motif.

The second item (ii) of the statement can be proven similarly following the ideas of~\cite[Theorem 6]{Bjorklund2016}.
That is, the nodes $s_{x,y,z}$ are duplicated $q$ times, i.e. into nodes $s_{x,y,z}^i, 1 \leq i \leq q$, where $c(s_{x,y,z}^i) = i$, forcing to have at most $q$ of such nodes in the solution.
Also, the $3q$ nodes $x_i$ receive a fresh unique color (say with colors $q+1$ to $q+1+3q$), forcing all of them to be in any solution.
The nodes $r_1,r_2,\dots,r_t$ are colored with color $q+1+3q+1$.

\end{proof}

\subsection{Parameters with unbounded neighborhood diversity}\label{subsec:unbounded-nd}

This section disproves the idea that \graphm is only tractable for classes with bounded neighborhood diversity.
Indeed, we show that \graphm is in $\fpt$ parameterized by the size of a \emph{vertex clique cover} or by the distance to co-cluster.
The former algorithm creates a win/win based on K\"onig's theorem applied to a bounded number of auxiliary bipartite graphs.
The latter is simpler and uses as subroutines the algorithms parameterized by vertex cover number and distance to clique.

In the \textsc{Vertex Clique Cover} problem (also known as \textsc{Clique Partition}), one asks, given a graph $G=(V,E)$ and an integer $k$, for a \textit{partition} of the vertices into $k$ subsets $C_1,\ldots,C_k \subseteq V$, such that $\forall i \in [k]$, $G[C_i]$ is a clique. 
The set $\{C_1,\ldots,C_k\}$ is called a \emph{vertex clique cover} of $G$.
The \emph{vertex clique cover number} of a graph $G$ is the smallest $k$ such that $G$ has an vertex clique cover of size $k$. 
This problem is equivalent to the \textsc{Graph Coloring} problem since a graph as a vertex clique cover of size $k$ if and only if its complement is $k$-colorable.
Therefore, \textsc{Vertex Clique Cover} is unlikely to be in $\xp$. 
However, if a vertex clique cover comes with the input, we show that \graphm is in $\fpt$ for parameter vertex clique cover number. 
One can notice that \graphm is $\np$-hard in $2$-colorable graphs.
This is a striking example of how easier can \graphm be on the denser counterpart of two complementary classes.

To realize that vertex clique cover number has unbounded neighborhood diversity, think of the complement of a bipartite graph.
The vertex clique cover is of size $2$ but the neighborhood diversity could be arbitrary;
for parameter distance to co-cluster, think of the complement of a cluster graph with an unbounded number of cliques.

\begin{theorem}\label{th:vertex-clique-cover}
\graphm can be solved in time $O^*(k^{O(k)})$ where $k$ is the vertex clique cover number, provided that the vertex clique cover is given as part of the input.
\end{theorem}
\begin{proof}

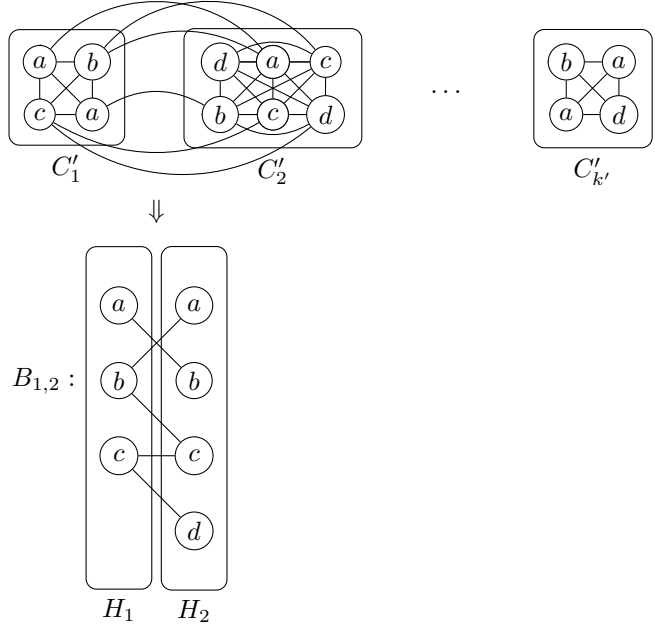
\begin{figure}[ht!]
\centering
\begin{tikzpicture}[scale=0.9,auto]

\node[vertex] (S11) {$a$};
\node[vertex, right of=S11,node distance=0.7cm] (S12)  {$b$};
\node[vertex, below of=S11,node distance=0.7cm] (S13)  {$c$};
\node[vertex, right of=S13,node distance=0.7cm] (S14)  {$a$};
\node[draw,rectangle, rounded corners,label=below:$C'_1$,fit= (S11) (S12) (S13) (S14)] (S1) {};
\foreach \i in {1,...,4} {
	\foreach \j in {\i,...,4} {
		\draw (S1\i) -- (S1\j);
	}
}

\node[vertex, right of=S12,node distance=1.7cm] (S21) {$d$};
\node[vertex, right of=S21,node distance=0.7cm] (S22)  {$a$};
\node[vertex, below of=S21,node distance=0.7cm] (S23)  {$b$};
\node[vertex, right of=S23,node distance=0.7cm] (S24)  {$c$};
\node[vertex, right of=S24,node distance=0.7cm] (S25)  {$d$};
\node[vertex, above of=S25,node distance=0.7cm] (S26)  {$c$};
\node[draw,rectangle, rounded corners,label=below:$C'_2$,fit= (S21) (S22) (S23) (S24) (S25) (S26)] (S2) {};
\foreach \i in {1,...,6} {
	\foreach \j in {\i,...,6} {
		\draw (S2\i) -- (S2\j);
	}
}
\node[vertex, fill=white,right of=S21,node distance=0.7cm] (S22)  {$a$};
\draw (S21) edge[bend left] (S26);
\node[vertex, fill=white,right of=S23,node distance=0.7cm] (S24)  {$c$};
\draw (S23) edge[bend right] (S25);

\node[right of=S26,node distance=1.6cm,yshift=-0.4cm] (dots) {$\ldots$};

\node[vertex, right of=S26, node distance=3.2cm] (Sk1) {$b$};
\node[vertex, right of=Sk1,node distance=0.7cm] (Sk2)  {$a$};
\node[vertex, below of=Sk1,node distance=0.7cm] (Sk3)  {$a$};
\node[vertex, right of=Sk3,node distance=0.7cm] (Sk4)  {$d$};
\node[draw,rectangle, rounded corners,label=below:$C'_{k'}$,fit= (Sk1) (Sk2) (Sk3) (Sk4)] (Sk) {};
\foreach \i in {1,...,4} {
	\foreach \j in {\i,...,4} {
		\draw (Sk\i) -- (Sk\j);
	}
}

\draw (S11) edge[bend left=50] (S22);
\draw (S12) edge[bend left=30] (S22);
\draw (S12) edge[bend left=50] (S26);
\draw (S14) edge[bend left=30] (S23);
\draw (S13) edge[bend right=30] (S24);
\draw (S13) edge[bend right=40] (S25);

\begin{scope}[scale=0.9,xshift=1.3cm,yshift=-4cm,auto]
\node at (0.62,1.55) {$\Downarrow$} ;
\node[vertex, inner sep=-0.26cm] (Ha) {$a$};
\node[vertex, below of=Ha] (Hb) {$b$};
\node[vertex, inner sep=-0.255cm, below of=Hb] (Hc) {$c$};
\node[inner sep=0.24cm, below of=Hc] (Hd) {};

\node[draw,rectangle, rounded corners,label=below:$H_1$,fit= (Ha) (Hb) (Hc) (Hd)] (H1) {};

\node[vertex, inner sep=-0.26cm, right of=Ha] (H2a) {$a$};
\node[vertex, below of=H2a] (H2b) {$b$};
\node[vertex, inner sep=-0.255cm, below of=H2b] (H2c) {$c$};
\node[vertex, below of=H2c] (H2d) {$d$};

\node[draw,rectangle, rounded corners,label=below:$H_2$,fit=(H2a) (H2b) (H2c) (H2d)] (H2) {};

\node[left of=Hb] () {$B_{1,2}:$};

\draw (Ha) -- (H2b);
\draw (Hc) -- (H2c);
\draw (Hc) -- (H2d);
\draw (Hb) -- (H2c);
\draw (Hb) -- (H2a);
\end{scope}

\end{tikzpicture}
\caption{The cliques $C'_1, C'_2, \ldots, C'_{k'}$, the edge interaction between $C'_1$ and $C'_2$, and the corresponding auxiliary bipartite graph $B_{1,2}$ when the multiset $M$ contains $a$ with multiplicity exactly one and $c$ with multiplicity at least $2$ (indeed, observe that the edge $cc$ is present in $B_{1,2}$ but not the edge $aa$).
}\label{fig:vcc}
\end{figure}

Let $(G=(V,E),c,M)$ be the instance and suppose that the partition into cliques $\{C_1,\ldots,C_k\}$ of the graph $G$ is given.
We remove all the vertices whose color does not belong to $M$, since they cannot be part of a solution. 
Observe also that this can only decrease the vertex clique cover number.
First, we guess in time $O^*(2^k)$ which of the cliques $\mathcal S=\{C'_1,\ldots,C'_{k'}\} \subseteq \{C_1,\ldots,C_k\}$ have a non-empty intersection with a fixed solution $R$, and we remove from $G$ the cliques which are not in $\mathcal S$.

We denote by $E(X,Y)$ the set of edges of $E$ having one endpoint in $X$ and the other in $Y$.
We call \emph{transversal edge} an edge in $E(C'_i,C'_j)$ with $i \neq j \in [k']$.
Such a transversal edge is said to have \emph{type} $\{i,j\}$.
An \emph{inner edge} is an edge which lies within the same clique $C'_i$ for some $i \in [k']$.
As $G[R]$ is connected, one may observe that there is a set $E_c \subseteq E(G[R])$ of $k'-1$ transversal edges such that between every pair of vertices $u$, $v \in R$, there is a path made only of edges in $E_c$ and inner edges.
Informally, $E_c$ is a spanning tree of the $k'$ cliques of $\mathcal S$ seen as vertices (see Figure~\ref{fig:vcc}).
More precisely, the edges of $E_c$ form a subforest of $G$.
We guess in time $O^*(k'^{2(k'-1)})$ the type of each edge in $E_c$.
We denote by $T_c$ the corresponding set of $k'-1$ types.

One may first think of the tansversal edges of $E_c$ as a matching.
Although two edges of $E_c$ leaving the same clique $C'_i$ can share the same vertex in $C'_i$.  
Actually this piece of information will prove useful for the algorithm to work.
Therefore, we also guess in time $O^*(B_{2(k'-1)})=O^*((2k')^{2k'})$ if two edges in $E_c$ of types $\{i,j\}$ and $\{i,j'\}$, happen to have a common endpoint.
One can see it the following way: among the potentially $2(k'-1)$ endpoints of the \emph{matching} $E_c$, we needed to find the correct partition into the classes of the equality relation. 
As $R$ is a solution, $M \subseteq c(C'_1 \cup \ldots \cup C'_{k'})$ holds.
Therefore, it all boils down to finding $k'-1$ transversal edges whose set of types is precisely $T_c$ and such that the multiset of colors of their at most $2(k'-1)$ endpoints is included in $M$.

For each type $\{i,j\} \in T_c$, we build the bipartite graph $B_{i,j}=(H_i \uplus H_j,F)$ where $H_i$ (resp. $H_j$) are all the colors of the vertices of $C'_i$ (resp. $C'_j$).
There is an edge in $F$ between color $c \in H_i$ and color $c' \in H_j$ whenever there is a transversal edge of type $\{i,j\}$ whose endpoint in $C'_i$ is colored by $c$ and whose endpoint in $C'_j$ is colored by $c'$.
In the special case when $c$ and $c'$ is in fact the same color \emph{and} that color appears only once in $M$, we remove the edge $cc'$ from $F$.
We indeed know that no solution will contain such a tranversal edge.
We remove all the isolated vertices of every $B_{i,j}$.
We also remove every vertex $c \in H_i$ from $B_{i,j}$ if there is a $j'$ such that we have guessed that the transversal edges of type $\{i,j\}$ and $\{i,j'\}$ share a common point and $c$ is not in the $H_i$ of $B_{i,j'}$ (it was an isolated vertex).
The rest of the algorithm is a win/win based on the classic K\"onig's theorem which states that, in a bipartite graph, the size of a minimum vertex cover is equal to the size of a maximum matching.
The core idea is that either there is a large diversity of colors for the endpoints of a transversal edge, and a suitable transversal edge can always be found at the end, or there is only a limited choice of colors for those endpoints and one can branch over those possibilities.   
By branching, we commit ourselves to find a transversal edge $uv$ whose endpoint, say, $u$ has a specific color $c$.  
In that case, we say that the endpoint $u$ has its color \emph{fixed}.
In a first step, we will branch until the endpoints of all the transversal edges are fixed (or can always be fixed). 
In a second step, we will build a solution respecting the fixed colors.  

We distinguish two cases.
Either, there is a matching $S_{i,j}$ in $B_{i,j}$ with at least $2k'-3$ edges.
Then, for any multiset of colors $M_o \subseteq M$ of size at most $2k'-4$, there is an edge $\{c,c'\}$ in $S_{i,j}$ such that $M_o \cup \{c,c'\} \subseteq M$.
Indeed, since $|S_{i,j}|>|M_o|$, there is at least one edge of $S_{i,j}$ whose endpoints are not colored by an element of $M_o$.
Recall also that there can be an edge between two vertices of the same color only if the multiplicity of that color in $M$ is at least $2$.
Therefore, whatever the multiset $M_o \subseteq M$ of colors at the endpoints of the $k'-2$ other transversal edges is, one can always find a transversal edge of type $\{i,j\}$ colored by $c$ and $c'$ such that $M_o \cup \{c,c'\} \subseteq M$.
Thus, we can forget about this particular transversal edge, and we say that the transversal edge of type $\{i,j\}$ is \emph{abundant}.

Otherwise, there is a vertex cover of $B_{i,j}$ with at most $2k'-4$ vertices. 
Note that a vertex $c \in H_i$ (resp.~$H_j$) in the graph $B_{i,j}$ corresponds to choosing color $c$ for the endpoint in $C'_i$ (resp.~$C'_j$) of the transversal edge of type $\{i,j\}$. 
Therefore, we branch on those at most $2k'-4$ possibilities of coloring one of the endpoints of the transversal edge of type $\{i,j\}$.

This describes what we do when no endpoint of the transversal edge has its color fixed.
Now, suppose we have a transversal edge of type $\{i,j\}$ such that the color of the endpoint in, say, $C'_i$ is fixed to color $c$.
If the number of neighbors of vertex $c \in H_i$ in the graph $B_{i,j}$ is at least $2k'-3$, we declare this edge abundant and no longer care about this edge.
Otherwise, if this number is at most $2k'-4$, we branch on the at most $2k'-4$ ways of coloring the endpoint in $C'_j$ of the transversal edge of type $\{i,j\}$.

Note also that when we fix the color of an endpoint in $C'_i$ of a transversal edge of type $\{i,j\}$, it also fixes the color of the endpoints in $C'_i$ of potential transversal edges of type $\{i,j'\}$ which we have guessed to share a common endpoint (in $C'_i$) with the transversal edge of type $\{i,j\}$.
Although, this potential set of transversal edges might very well be empty.
After a branching of depth at most $2k'-2$ and arity at most $2k'-4$, we reach a situation where each transversal edge is either abundant or both its endpoints have fixed colors. 
We fix the colors of the endpoints of the abundant transversal edges (which are not fixed yet) in the following way.
For each tree of the forest $E_c$, we root them arbitrarily.
We then consider an arbitrary parent of some deepest leaves.
We fix the colors of the endpoints corresponding to this parent and all its children.
We explained above why this is always possible.
We iterate this until every vertex of this tree has its color fixed.

Now, all the endpoints of the transerval edges have their color fixed.
By guessing the set $T_c$ of types of the transversal edges and whether or not two transversal edges are incident, we have in fact guessed the shape of a forest that those edges constitute in the original graph $G$.
For each tree of this abstract forest, we have to compute the actual transversal edges.
At this point, a node in this tree is naturally labeled by a pair (clique,color) $(C'_i,c)$.
We associate a subset of vertices to a node of this labeled tree in a bottom-up fashion.
Each leaf labeled by $(C'_i,c)$ is associated with the subset $J_{i,c} \subseteq C'_i$ of vertices colored by $c$ (that is, $\forall u \in C'_i$, $u \in J_{i,c} \Leftrightarrow c(u)=c$).
We associate each inner node labeled by $(C'_i,c)$ whose $r$ children are associated with sets $J_{i_1,c_1}, \ldots, J_{i_r,c_r}$ with the subset $J_{i,c} \subseteq C'_i$ of vertices colored by $c$ which have at least one neighbor in $J_{i_h,c_h}$ for each $h \in [r]$. 
When the last node $e$ of the tree gets its set $J$, this set is non empty if we have made all our guesses accordingly to solution $R$.
We define $e$ as the root of the tree.
Now, in a top-down manner we find the corresponding transversal edges.
We take in the solution an arbitrary vertex $u \in J$.
In each set associated with a child of $e$ we take arbitrarily a neighbor of $u$; and so on, up to the leaves.
By construction, this is always possible.
It is possible that while doing this process on two different trees of the forest, we take "twice" the same vertex in some $C'_i$.
This can only help since the goal is not to exceed the multiplicities of $M$.
Equivalently, we could have guessed the forest of transversal edges with the least number of connected components, to forbid this possibility.

We summarize the algorithm.\\
1) Guess the shape of the forest formed by a fixed subset $E_c$ of $k'-1$ transversal edges ensuring the connectivity between the cliques in a fixed solution $R$.\\
2) Win/win to properly guess the colors of the endpoints of $E_c$: (a) either the variety of colors is more than enough and this color can be fixed arbitrary later, or (b) the are only few choices and one can branch.\\
3) For each tree of $E_c$, find the transversal edges: one bottom-up procedure to check if there is indeed a solution and one top-down to select the actual vertices.\\ 
4) As $R$ is a solution, one can complete this to a solution by taking arbitrary (since everything is connected) vertices with the right colors. 

Observe that during step 2), we first do all the branchings advocated by (b).
Then we reach a point when no further branching is possible, and we fix the colors arbitrarily as indicated by (a).

The running time of the algorithm is $O^*(2^kk^{2k-2}(2k)^{2k}(2k-4)^{2k-2})=O^*((4\sqrt{2}k)^{6k})=O^*(k^{O(k)})$.
\end{proof}



\begin{theorem}\label{th:cocluster}
\graphm can be solved in $O^*(2^{2k \log k})$, where $k$ is the distance to co-cluster.
\end{theorem}
{
\begin{proof}
Let $(G=(V,E),c,M)$ be any instance of \graphm and let $R$ be a solution.
Let $X$ be a minimum subset (of size $k$) whose deletion makes the graph $G$ a co-cluster. 
Co-cluster graphs are exactly the $\overline{P_3}$-free graphs.  A $P_3$ graph is a path with 3 nodes (a $\overline{P_3}$ graph is its complement, thus one node and one edge).
We can apply a bounded-depth branching algorithm by finding a $\overline{P_3}$ and branching on which of the three vertices to put into the solution. 
This leads to an $O^*(3^k)$ algorithm to find $X$.
Let $S_1,S_2,...,S_q$ be the partition of the co-cluster graph $G[V \setminus X]$ into maximal independent sets.
The idea is to run the algorithm parameterized by the vertex cover number if at most one $S_i$ is inhabited by solution $R$, and the one parameterized by distance to clique otherwise.
Therefore, we distinguish two cases: 

\begin{enumerate}[(A)]
\item  $|\{i \in [q]$ $|$ $R \cap S_i \neq \emptyset\}| \leqslant 1$, \label{cA}
\item $|\{i \in [q]$ $|$ $R \cap S_i \neq \emptyset\}| \geqslant 2$.\label{cB}
\end{enumerate}

In case (\ref{cA}) holds, we will find a solution by solving, for each $i \in [q]$, the instance $(G[X \cup S_i],c_{|X \cup S_i},M)$.
As $X$ is a vertex cover of size $k$ in $G[X \cup S_i]$, this can be done in time $O^*(2^{2k \log k})$ by Theorem~\ref{th:vcfpt}.

In case  (\ref{cB}) holds, we can guess in time $n^2$ one vertex $s \in S_i \cap R$ and one vertex $t \in S_j \cap R$ with $i \neq j \in [q]$.
Then, we will find a solution by solving $(G'=(V \setminus \{s,t\},E'),c_{V \setminus \{s,t\}},M \setminus c(\{s,t\}))$ where $E'=(E \cup \{\{u,v\}$ $|$ $u,v \in S_a, a \in [q]\})_{|V \setminus \{s,t\}}$.
Indeed, if $Y \subseteq V \setminus \{s,t\}$ induces a connected subgraph in $G'$, then $G[Y \cup \{s,t\}]$ is connected.
As $G'-X$ is now a clique, this can be done in time $O^*(3^k)$ by Theorem~\ref{th:distclique}.
The overall running time is $O^*(3^k + q 2^{2k \log k}+n^2 3^k)=O^*(2^{2k \log k})$.
\end{proof}
}




\subsection{ETH-based lower bounds}

Here, we show that a parameterized subexponential algorithm (i.e., running in $O^*(2^{o(k)})$) solving \graphm for the parameters $k$ that we considered in this section, is unlikely.
We get those negative results as a corollary of the fact that, while trying out all the subsets of vertices obviously solves \graphm in time $O^*(2^n)$, a subexponential time algorithm (i.e., running in $2^{o(n)}$) is unlikely:

\begin{theorem}\label{thm:lower-bound-eth-main}
Under ETH, \graphm cannot be solved in time $2^{o(n)}$, even (i) on graphs with distance $1$ to cluster, and (ii) on trees.
\end{theorem}

\begin{proof}
Under ETH, \textsc{Dominating Set} restricted to graphs with degree $6$ is not solvable in time $2^{o(n)}$ where $n$ is the number of vertices of the input graph \cite{Fomin04}.
From a degree-$6$ graph $H$ and an integer $t$, we build an instance $\mathcal I=(G=(V,E),c:V \rightarrow \mathcal C,M)$ of \graphm such that there is a dominating set of size at most $t$ in $H$ if and only if $\mathcal I$ is a YES-instance.
First we show item (i).
There are $|V(H)|+1$ different colors in $\mathcal C$, one color $c_v$ for each vertex $v$ of $H$, and one special color $c$.
For each vertex $v$ in $H$, we introduce a clique in $G$ of size $|N_H[v]|$ ($\leqslant 8$) where one vertex is colored by the special color $c$, and the others are colored by each color of $\{c_w | w \in N_H[v]\}$.
We add a vertex $z$ colored by $c$ and link it to all the other vertices colored by $c$ (in the cliques).
The motif $M$ consists of $c$ with multiplicity $t+1$ and $c_v$ (for each $v \in V(H)$) with multiplicity $1$.
That ends the construction.
Observe that the number of vertices of $G$ is linear in $|V(H)|$ (it is at most $8|V(H)|+1$), and removing $z$ from $G$ gives a cluster graph of $|V(H)|$ cliques of size at most $8$ each.

To obtain item (ii), $G$ is transformed in the following way: each clique is replaced by a star where the center is the vertex with the special color $c$.

Those reductions are identical to the reduction showing that \graphm is hard on trees of diameter $4$ \cite{Ambalath2010} (for (ii)) and Theorem~\ref{th:cluster} (for (i)), and therefore the reader is referred to paper~\cite{Ambalath2010} for correctness.
\end{proof}

\begin{corollary}
Under ETH, for every parameter upper-bounded by $n$, \graphm cannot be solved in time $2^{o(k)}$, even on trees.
\end{corollary}

Among the six parameters for which we gave an FPT algorithm, two are not upper-bounded by $n$ but by $n^2$: cluster editing and edge clique cover numbers.
Though, we can observe that the graph built in item (i) of the proof of Theorem~\ref{thm:lower-bound-eth-main} has both a cluster editing of size $n$ (by removing the $n$ edges between $z$ and the $n$ other vertices colored by $c$) and an edge clique cover of size $2n$.
Therefore, for all the six parameters, a subexponential parameterized algorithm in $2^{o(k)}$ would disprove ETH.

We finally show finer lower bounds under SETH and SCH, for parameter vertex cover and distance to clique.
In particular, Theorem~\ref{th:schdtc} implies that, even though there should be an algorithm solving \graphm in time $O^*(c^k)$ with $c<8$, and $k$ being the distance to a clique, (thereby, improving over Theorem~\ref{th:distclique}), it is unlikely that $c$ goes below $2$. 

\begin{theorem}\label{th:sethvc}
Under SETH, for any $\varepsilon > 0$, \graphm cannot be solved in time $O((2-\varepsilon)^k)$, where $k$ is the vertex cover number.
\end{theorem}

\begin{proof}
In the \textsc{Hitting Set} problem, one is given a set of sets $\mathcal S=\{S_1, \ldots, S_m\}$ over elements $\mathcal X=\{x_1, \ldots, x_n \}$, and an integer $t$, and one has to find a set $\mathcal X' \subseteq \mathcal X$ (the hitting set) of size at most $t$ such that $\forall S \in \mathcal S, S \cap \mathcal X' \neq \emptyset$.
It is known that under SETH, for any $\varepsilon>0$, \textsc{Hitting Set} is not solvable in time $O((2-\varepsilon)^n)$~\cite{CyganD12}.
From any instance $(\mathcal X, \mathcal S, t)$ of \textsc{Hitting Set} with $n$ elements, we construct an equivalent instance $(G=(V,E),c,M)$ of \graphm where the graph $G$ has a vertex cover of size $n$.
We create one vertex $v(x_i)$ for each element $x_i$ of $\mathcal X$ and one vertex $v(S_j)$ for each set $S_j$ of $\mathcal S$.
The \emph{element} vertices (the $v(x_i)$'s) are colored by $1$ and form a clique, while the \emph{set} vertices (the $v(S_j)$'s) are colored by $2$ and constitute an independent set.
We link an \emph{element} vertex to a \emph{set} vertex if the corresponding element is in the corresponding set; that is, $v(x_i)v(S_j) \in E \Leftrightarrow x_i \in S_j$.
Therefore, $G$ is the adjacency split graph of the set-system $(\mathcal X,\mathcal S)$ where the \emph{element} vertices are the clique.
$M$ contains $1$ with multiplicity $t$ and $2$ with multiplicity $m$. 
Observe that the set of all the \emph{element} vertices is a vertex cover of $G$ of size $n$.

If $\mathcal X'=\{x_{a_1}, \ldots, x_{a_t}\}$ is a solution (potentially, add arbitrary elements to get a solution with \emph{exactly} $t$ elements) to the hitting set instance, then $R:=\bigcup_{j \in [m]}v(S_j) \cup \{v(x_{a_1}), \ldots, v(x_{a_t})\}$ (obtained by taking all the \emph{set} vertices and the $t$ \emph{element} vertices corresponding to the elements of $\mathcal X'$) satisfies the multiset constraint. 
Also, the subgraph $G[R]$ is indeed connected by the definition of a hitting set, and the fact that $\{v(x_{a_1}), \ldots, v(x_{a_t})\}$ is a clique. 

Conversely, let $R \subseteq V$ be a solution for the constructed instance of \graphm. 
By the multiset constraint, $R$ should contain all the vertices colored by $2$: $\bigcup_{j \in [m]}v(S_j)$, and $t$ vertices colored by $1$: $\{v(x_{a_1}), \ldots, v(x_{a_t})\}$.
We claim that $\mathcal X' := \{x_{a_1}, \ldots, x_{a_t}\}$ is a hitting set (of size $t$).
Indeed, if a set $S_j$ was not hit by $\mathcal X'$, then the \emph{set} vertex $v(S_j)$ would not be connected to the clique $\{v(x_{a_1}), \ldots, v(x_{a_t})\}$, and $G[R]$ would have at least $2$ connected components.
\end{proof}

\begin{theorem}\label{th:schdtc}
Under SCH, for any $\varepsilon > 0$, \graphm cannot be solved in time $O((2-\varepsilon)^k)$, where $k$ is the distance to clique.
\end{theorem}

\begin{proof}
From an instance of \textsc{Set Cover} with $n$ elements, we build an equivalent instance of \graphm where the distance from the graph to a clique is $n$.
Again, we create one vertex for each element and one vertex for each set.
The \emph{element} vertices are colored by $1$ and constitute an independent set, while the \emph{set} vertices are colored by $2$ and form a clique.
We link each \emph{element} vertex to each \emph{set} vertex if the corresponding element is in the corresponding set.
The graph is the adjacency split graph where the \emph{set} vertices are the clique.
$M$ contains $1$ with multiplicity $n$ and $2$ with multiplicity $t$. 
The removal of the set of all the \emph{element} vertices (of size $n$) would leave a clique. 
The correctness of the reduction is similar to the one of Theorem~\ref{th:sethvc}.
\end{proof}


\section{Parameters for which \graphm is hard}\label{sec:hard}

In this section, we provide several parameters for which \graphm is not in $\xp$, unless $\p = \np$. 
In other words, the problem is $\np$-hard even for fixed values of the parameter. 
We also prove that the problem remains $\wone$-hard for parameter max leaf number. 
Figure~\ref{fig:recap} summarizes these results.



\subsection{Deletion set numbers}

We study parameters which correspond to the minimum number of vertices to remove to make the graph belong to a restricted class.
We will show that \graphm remains $\np$-hard for constant values of those parameters. 
More precisely, the colorful restriction of \graphm is hard even if we can obtain a set of disjoint paths by removing $1$ vertex, a cluster graph by removing $1$ vertex, and an acyclic graph by removing $0$ edge. 

\begin{theorem}[\cite{fellows2011}] \graphm is $\np$-hard even when $G$ is a tree of maximum degree 3 and the motif is colorful.
\end{theorem}

\begin{corollary}
\graphm is $\np$-hard even for graphs with feedback edge set number 0 and when the motif is colorful.
\end{corollary}

\begin{theorem}\label{thm:disjoint-paths}
\graphm is $\np$-hard even (i) for graphs with distance $1$ to disjoint paths and when the motif is colorful and (ii) for graphs with bandwidth $6$ and when the motif is colorful.
\end{theorem}
{
\begin{proof}
We will detail only (i). We propose a reduction from \textsc{Exact Cover by 3-Sets} (\pxtc). This special case of \textsc{Set Cover} is known to be $\np$-complete. Recall that \pxtc is stated as follows, given a set $X = \{x_1,x_2,\dots,x_{3q}\}$ and a collection $\S = \{S_1,\dots,S_{|\S|}\}$ of 3-elements subsets of $X$, the goal is to decide if $\S$ contains a subcollection $\T \subseteq \S$ such that each element of $X$ occurs in exactly one element of $\T$. The size of $X$ must be a multiple of three since a solution is a set of triplets where each element of $X$ must appear exactly once.

Let us now describe the construction of an instance $\I'=(G=(V,E),c,M)$ of \graphm from an arbitrary instance $\I = (X,\S)$ of \pxtc (see also Figure~\ref{fig:X3Cconstruction}). The graph $G=(V,E)$ is built as follows: there is a distinct root $r$, for each $S_i \in \S$, there are two paths built from $r$, the first one is made of a node $a_i^1$, three nodes representing the elements in $S_i$ and a node $b_i^1$, the other one is made of two nodes $a_i^2$ and $b_i^2$. The graph is thus a tree such that removing $r$ gives a collection of $2|\S|$ paths.

The set of colors is $\C = \{ 1,2,\ldots, 2|\S|+3q+1 \} $. The coloration of $G$ is such that $c(a_i^1) = c(a_i^2) = i$ and $c(b_i^1) = c(b_i^2) = |\S|+i$ for $1 \leq i \leq |\S|$, the $3q$ colors $2|\S|+1,\ldots,2|\S|+3q$ are assigned to vertices corresponding to $X$, and $c(r) = 3q+2|\S|+1$. The motif is equal to the set of colors and is thus colorful. This construction is clearly done in polynomial time in regards of $\I$.

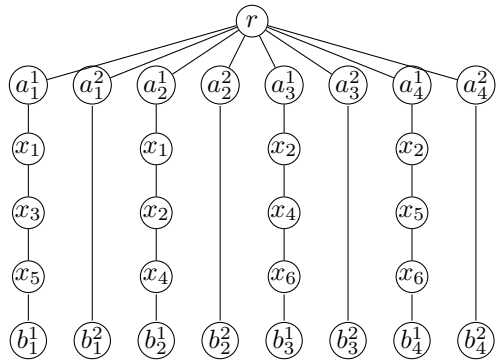
\begin{figure}[ht!]
\begin{center}
\begin{tikzpicture}[scale=.85,>=stealth,shorten <=.5pt,shorten >=.5pt]
\tikzstyle{vertex}=[circle,fill=black!0,minimum size=12pt,inner sep=0pt]

	\node[vertex,draw] (r) at (3.5,2) {$r$};
	\node[vertex,draw] (a11) at (0,1) {$a_1^1$};
	\node[vertex,draw] (11) at (0,0) {$x_1$};
	\node[vertex,draw] (13) at (0,-1) {$x_3$};
	\node[vertex,draw] (15) at (0,-2) {$x_5$};
	\node[vertex,draw] (b11) at (0,-3) {$b_1^1$};
	\node[vertex,draw] (a12) at (1,1) {$a_1^2$};
	\node[vertex,draw] (b12) at (1,-3) {$b_1^2$};
	\node[vertex,draw] (a21) at (2,1) {$a_2^1$};
	\node[vertex,draw] (21) at (2,0) {$x_1$};
	\node[vertex,draw] (22) at (2,-1) {$x_2$};
	\node[vertex,draw] (24) at (2,-2) {$x_4$};
	\node[vertex,draw] (b21) at (2,-3) {$b_2^1$};
	\node[vertex,draw] (a22) at (3,1) {$a_2^2$};
	\node[vertex,draw] (b22) at (3,-3) {$b_2^2$};
	\node[vertex,draw] (a31) at (4,1) {$a_3^1$};
	\node[vertex,draw] (32) at (4,0) {$x_2$};
	\node[vertex,draw] (34) at (4,-1) {$x_4$};
	\node[vertex,draw] (36) at (4,-2) {$x_6$};
	\node[vertex,draw] (b31) at (4,-3) {$b_3^1$};
	\node[vertex,draw] (a32) at (5,1) {$a_3^2$};
	\node[vertex,draw] (b32) at (5,-3) {$b_3^2$};
	\node[vertex,draw] (a41) at (6,1) {$a_4^1$};
	\node[vertex,draw] (42) at (6,0) {$x_2$};
	\node[vertex,draw] (45) at (6,-1) {$x_5$};
	\node[vertex,draw] (46) at (6,-2) {$x_6$};
	\node[vertex,draw] (b41) at (6,-3) {$b_4^1$};
	\node[vertex,draw] (a42) at (7,1) {$a_4^2$};
	\node[vertex,draw] (b42) at (7,-3) {$b_4^2$};
	
	\path[draw] 
      		(r) -- (a11) -- (11) -- (13) -- (15) -- (b11);
	\path[draw] 
      		(r) -- (a21) --  (21) -- (22) -- (24)-- (b21);
	\path[draw] 
      		(r) -- (a31) -- (32) -- (34) -- (36)-- (b31);
	\path[draw] 
      		(r) -- (a41) -- (42) -- (45) -- (46)-- (b41);
	\path[draw] 
      		(r) -- (a42) -- (b42);
	\path[draw] 
      		(r) -- (a32) -- (b32);
	\path[draw] 
      		(r) -- (a22) -- (b22);
	\path[draw] 
      		(r) -- (a12) -- (b12);
\end{tikzpicture}
\end{center}
\caption{The graph $G$ built from $X = \{x_1,x_2, \dots ,x_6\}$ (thus with $q=2$) and $\S = \{\{x_1,x_3,x_5\},\{x_1,x_2,x_4\},\{x_2,x_4,x_6\},\{x_2,x_5,x_6\}\}$.}\label{fig:X3Cconstruction}
\end{figure}

Let us now prove that if there is a solution for an instance $\I$ of \pxtc, then there is solution for the instance $\I'$ of \graphm. Given a solution $\T \subseteq \S$ for $\I$, a solution $P$ for $\I'$ is built as follows: take the root, for each $S_i \in \T$, take the whole path from $a_i^1$ to $b_i^1$, and for each $S_i \notin \T$, take the path $a_i^2b_i^2$. Informally speaking, for each set, either the set is in $\T$ and thus the path with the nodes corresponding to the elements is taken, otherwise the path with only two nodes is taken. By definition of a solution for $\I$, each color $2|\S|+1,\ldots,2|\S|+3q$ is taken only once, and for each color $1,\ldots,2|\S|$, exactly one of the two occurrences is taken. The root is also taken and thus the solution is connected.

Conversely, let us now prove that there is a solution for the instance $\I$ of \pxtc if there is a solution for the instance $\I'$ of \graphm. First observe that the root $r$ must be in the solution since it is the only node with this color. Also, for each $1 \leq i \leq |\S|$, either $a_i^1$ or $a_i^2$ must be in the solution since it is the only node with color $i$. The same holds for $b_i^1$ and $b_i^2$. Also, observe that if $a_i^1$ is in the solution, then $b_i^1$ must also be in the solution, with the three element nodes along the path. Indeed, if it is not the case, the color $c(b_i^1)$ will never be in the solution since the only other node with this color is $b_i^2$. However, in order to add $b_i^2$ in the solution, $a_i^2$ must be in the solution to respect the connectivity constraint, which is impossible since $c(a_i^1) = c(a_i^2)$. Therefore, either the three element nodes corresponding to a set $S_i \in \S$ are entirely in the solution $P$, or none are.
The solution is built as follows: $\T = \{S_i : a_i^1 \in P \}$. Since $P$ is a solution, colors of $P$ appear exactly once. Therefore, each element of $X$ appears exactly once in $\T$.

For (ii), we slightly modify the graph $G$. 
Instead of having one vertex $r$ linked to each $a_i^j$ (for $i \in [|\mathcal S|]$ and $j \in [2]$), we now have a path $R=r_1^1r_1^2r_2^1r_2^2 \ldots r_{|\mathcal S|}^1r_{|\mathcal S|}^2$, and for each $i \in [|\mathcal S|]$ and $j \in [2]$, there is an edge between $r_i^j$ and $a_i^j$.
We call that new graph $H$. 
We may observe that $H$ is a comb graph whose spine is $R$.
The set of colors is now $\mathcal C=[4|\mathcal S|+3q]$.
All the vertices in $G - r$ keep the same colors, and for each $i \in [|\mathcal S|]$ and $j \in [2]$, $c(r_i^j)=2|\mathcal S|+3q+2(i-1)+j$.
In other words, we give a fresh and distinct color to each vertex of $R$.
Again, the motif $M$ is the entire set of colors $\mathcal C$. 
The correctness is the same as for (i), since all vertices of $R$ must be in any solution because they are the only occurrences of their respective color.
Since the maximal paths having exactly one vertex in the spine $R$, called \emph{teeth}, are of length at most $6$, the bandwidth of $H$ is bounded by $6$, too.
Indeed, one can number the vertices increasingly tooth by tooth.
A more careful analysis shows that the bandwidth of $H$ is actually $5$.
\end{proof}

Actually, one could also follow the reduction of \cite{Cygan12} but start from a version of \textsc{Sat} where each literal appears in at most two clauses.
This variant is also $\np$-complete, and the graph produced would have bandwidth $4$. 

}

\begin{theorem}\label{th:cluster}
\graphm is $\np$-hard even for graphs with distance $1$ to cluster and when the motif is colorful.
\end{theorem}
{
\begin{proof}
To prove this theorem, one can use the reduction from \textsc{Colorful Set Cover} to \graphm where the input graph is a tree of diameter at most 4 (called superstar)~\cite{Ambalath2010}. 
The idea is just to replace each subtree representing a set $S_i$ by a clique of size $|S_i|+1$. 
Removing the root of the former superstar in this new graph yields a disjoint union of cliques and the rest of the proof carries over.
\end{proof}
}


\subsection{Dominating set number}

Being given a small dominating set of the graph cannot help in solving \graphm.
For any instance $(G=(V,E),c,M)$, one may add a universal new vertex $v$ to $G$, and color it with a color which does not appear in motif $M$.
The minimum dominating set $\{v\}$ is of size~$1$.
Vertex $v$ cannot be part of the solution due to its color, so answering the new problem is as hard as solving the original instance.
However, this could be considered as cheating since a vertex whose color is not in $M$ can immediately be discarded from the graph.
We show that even when $\forall v \in V$, $c(v) \in M$, graphs with dominating set of size $2$ can be hard to solve.

\begin{theorem}\label{th:ds}
\graphm is $\np$-hard even for graphs with a minimum dominating set of size~$2$ and when the motif is colorful.
\end{theorem}
{
\begin{proof}
We reduce from a rooted variant of \graphm, where the solution should contain a special vertex $r$.
This variant was proven $\np$-hard by Ambalath et al.~\cite{Ambalath2010}.

We will now prove that the problem remains hard with a small dominating set. The informal idea is to add a universal node $u$ such that the dominating set is small, but with a gadget to avoid the possibility of having this universal node in a solution (making the problem easy since any subset will be connected due to $u$). More formally, from any instance $\mathcal I=(G=(V,E),c,M)$, and any fixed vertex $r$ in $V$, we build the instance $\mathcal I'=(G'=(V \cup \{u,s,t\},E'),c',M')$, where $E'=E \cup \{\{s,t\},\{t,r\}\} \cup \{\{u,w\}$ $|$ $w \in V\}$, $c'(w)=c(w)$ for each $w \in V$, $c'(t)=c'(u)=x$, $c'(s)=y$, with $x$ and $y$ being two distinct fresh colors, and $M'=M \cup \{x,y\}$. 
By construction, $\{u,t\}$ is a dominating set in $G'$ of size $2$.
Let $R$ be a solution of \graphm for instance $\mathcal I'$.
Vertex $s$ is the only vertex with color $y$, so it has to be in $R$.
But then, as the only neighbor of $s$ is $t$ (and $|M'| \geqslant 2$), $t$ should also be in $R$.
Only one vertex with color $x$ can be in $R$, so $u$ cannot be part of the solution.
Now, the problem is as hard as solving instance $\mathcal I$ rooted in $r$. 
\end{proof}
}

\subsection{Max leaf number}

The \emph{max leaf number} of a graph $G$, denoted by $\text{ml}(G)$, is the maximum number of leaves (i.e., vertices of degree $1$) in a spanning tree of $G$. 
Therefore, if $G$ is itself a tree, $\text{ml}(G)$ is simply the number of leaves of $G$.
We will first show that \graphm is in $\xp$ parameterized by max leaf number.
The $n^{O(\text{ml(G)})}$ running time of our algorithm relies on a simple structural lemma that we state here:

\begin{lemma}\label{lem:max-leaf-easy}
Let $G=(V,E)$ be a connected graph and $S \subseteq V$ be the subset of all the vertices of $G$ of degree at least $3$.
Then $|S| \leqslant 4\text{ml}(G)$ and $G[V \setminus S]$ is a disjoint union of at most $5\text{ml}(G)$ paths.
\end{lemma}

\begin{proof}
The first part of the lemma (\emph{$|S| \leqslant 4ml(G)$}) is already known \cite{Kleitman91}.
Let us now prove the second part: \emph{$G[V \setminus S]$ is a disjoint union of at most $5\text{ml}(G)$ paths}.

As $G$ is connected, we can find $s-1$ paths $P_1, \ldots, P_{s-1}$ of $G[V \setminus S]$ such that $G[S \cup P_1 \cup \ldots \cup P_{s-1}]$ is connected, where $s$ is the number of connected components of $G[S]$.
Therefore, we build the following spanning tree of $G$: we start by taking the edges of any spanning forest of $G[S]$, plus all the edges incident to at least one vertex of a path $P_i$ (for $i \in [s-1]$).
Now, all the remaining paths in $G[V \setminus S]$ will provide (at least) one leaf each.
As $s \leqslant |S| \leqslant 4k$, if the number of paths in $G[V \setminus S]$ were larger than $5k$, then we could exhibit a spanning tree with at least $k+1$ leaves, which is a contradiction to $k=\text{ml}(G)$.
\end{proof} 

On the negative side, we will prove that \graphm is $\wone$-hard with parameter max leaf number, which is to the best of our knowledge, the first problem to exhibit such a behavior.
In fact, we will even prove that it is $\wone$-hard on trees with parameter \emph{number of leaves in the tree plus number of distinct colors in the motif}. 
This strengthens the previously known result that the problem is $\wone$-hard on trees with parameter number of distinct colors in the motif~\cite{fellows2011}.

\begin{theorem}\label{thm:max-leaf-easy}
\graphm can be solved in time $O^*(16^kn^{10k})=n^{O(k)}$, where $k=\text{ml}(G)$ and is in $\wpp$ with respect to that parameter.
\end{theorem}

{
\begin{proof}
Let $(G=(V,E),c,M)$ be any instance of \graphm, $k=\text{ml}(G)$, and $S$ the set of vertices with degree strictly greater than $2$ in $G$.
Again, we may assume that $G$ is connected and also that $G$ is not a cycle, since otherwise \graphm is trivially solvable in time $O(n^2)$.

It is known that $|S| \leqslant 4k$ (even $4k-2$) \cite{Kleitman91}.
First, we can exhaustively find in time $2^{4k}=16^k$ the intersection $T=S \cap R$, where $R$ is a fixed solution.
By definition, $V \setminus S$ are vertices of degree at most $2$. 
In particular, $G[V \setminus S]$ is a disjoint union of paths (some of the paths may consist of a single vertex). 
Indeed, there cannot be a cycle in $G[V \setminus S]$ since this cycle could not be connected to the rest of $G$.
By Lemma \ref{lem:max-leaf-easy}, the number of paths in $G[V \setminus S]$ is at most $5k$.

To satisfy the connectivity constraint, solution $R$ can intersect each of the at most $5k$ paths of $G[V \setminus S]$ in at most $n^2$ different ways (more precisely in at most ${l \choose 2}+l+1$ where $l$ is the number of vertices in the path).
So, we can guess the intersection $R \cap (V \setminus S)$ in time $(n^2)^{5k}=n^{10k}$.
Overall, we can decide \graphm in time $O^*(16^kn^{10k})$ where $k$ is the max leaf number.

We can also show that \graphm parameterized by $\text{ml}(G)$ is in $\wpp$ with the characterization of this class by Turing machines with bounded non-determinism \cite{Cesati03}.
\end{proof}
}

\begin{theorem}\label{thm:max-leaf-hard}
\graphm is $\wone$-hard with respect to the max leaf number plus the number of colors, even on trees.
\end{theorem}

{
\begin{proof}
We show the stronger result that \graphm is $\wone$-hard on subdivisions of the star $K_{1,k}$ with parameter $k+|\mathcal C|$ where $\mathcal C$ is the set of colors.
From any instance $H=(H_1 \uplus \ldots \uplus H_k,E)$ of the $\wone$-hard problem \textsc{Multicolored $k$-Clique}, we construct an equivalent instance $(T=(V,E'),c:V \rightarrow \mathcal C,M)$ of \graphm where $T$ is a tree with $k+{k \choose 2}+1$ leaves and $\mathcal C$ consists of ${k \choose 2}+3$ colors.
More precisely, $T$ is a subdivision of the star with $k+{k \choose 2}+1$ leaves.
We recall that the \textsc{Multicolored $k$-Clique} problem asks for a $k$-clique in $H$ hitting each $H_i$ (exactly once).
By potentially adding some isolated vertices, we can assume that each $H_i$ contains the same number $t$ of vertices, and $H_i=\{u_{i,1}, \ldots, u_{i,t}\}$.

The set of colors $\mathcal C$ is $\{c_0,c_b,c_e\} \cup \bigcup_{i < j \in [k]}\{ij\}$ ($|\mathcal C|={k \choose 2}+3$).
The motif $M$ contains $c_0$ with multiplicity $1$, both $c_b$ and $c_e$ with multiplicity $s:=k(t-1)+{k \choose 2}t^2$, and for any $i < j \in [k]$, color $ij$ with multiplicity $t^2$.
We write $M=\{1 \times c_0, s \times c_b, s \times c_e\} \cup \bigcup_{i<j \in [k]}\{t^2 \times ij\}$, with the convention that $\text{mul} \times \text{col}$ means color $\text{col}$ appears in the multiset with multiplicity $\text{mul}$.

The tree $T$ is a subdivision of a star with $k+{k \choose 2}+1$ leaves whose center $v$ is the only vertex colored by $c_0$.
Thus, $v$ should necessarily be in any solution.
By construction, $T[V \setminus \{v\}]$ is a disjoint union of $k+{k \choose 2}+1$ paths. 
We can think those paths as \emph{oriented} from the vertex neighbor of $v$ (the \emph{first} vertex of the path) to the vertex the farther away from $v$ (the \emph{last} vertex of the path).
We will extensively call those paths \emph{oriented paths}. 
By this, we only mean something informal about a potential solution growing from $v$ along those paths, and we do \emph{not} mean that the graph we build is directed.
For each $i \in [k]$, a path $P_i$ will correspond to the vertices of $H_i$ and for any pair $i<j \in [k]$, a path $P_{i,j}$ will encode the edges of $E_{i,j}:=E(H_i,H_j)$.
Additionally, we have a path $P_{be}$ with $2s$ vertices alternating color $c_b$ and $c_e$; the first vertex of the path is colored by $c_b$, the second by $c_e$ and so forth.

Before we describe the $P_i$s and the $P_{i,j}$s, we introduce the notion of \emph{block} and indicate a useful property that the construction will satisfy.
A \emph{block} is a subpath of an oriented path which starts with a vertex colored by $c_b$ (as \textbf{b}egin), ends with a vertex colored by $c_e$ (as \textbf{e}nd), and such that no internal vertex in the subpath has color $c_b$ or $c_e$.
The path $P_{be}$ can be seen as $s$ consecutive \emph{empty} \bls.
We may also observe that two different \bls of the same oriented path cannot intersect.
We will construct the $P_i$s and the $P_{i,j}$s such that they are entirely spanned by \bls; and we call that \emph{alternating property}.
Therefore, every vertex except $v$ is contained in a (unique) \bl.
In particular, each oriented path $P_i$ or $P_{i,j}$ has its first vertex colored by $c_b$ and its last vertex colored by $c_e$.
And, if we only consider vertices colored by $c_b$ and $c_e$ along the path, they alternate $c_b-c_ec_b-c_e \ldots$ with the extra property that there is no vertex between color $c_e$ and $c_b$ (see Figure~\ref{fig:ml-hardness1}).
A connected subgraph of $T$ containing $v$ (i.e., a potential solution) is entirely defined by $k+{k \choose 2}+1$ \emph{stopping points}: one for each oriented path $P_{be}$, $P_i$, or $P_{i,j}$. 
A \emph{stopping point} of an oriented path $P$ with respect to a given (attempt of) solution $R$ is the farthest vertex from $v$ lying in $R \cap P$. 
Observe that the unique path from $v$ to a stopping point is exactly the intersection of the solution and the oriented path.
If $R \cap P = \emptyset$, by convention, the stopping point is $v$.
It is easy to see that, in each oriented path $P_{be}$, $P_i$, or $P_{i,j}$, a stopping point relative to an actual solution is either $v$ or a vertex colored by $c_e$ (that is the end of a \bl). 
Put differently, if $R$ is a solution and $B$ is a \bl, $R \cap B = \emptyset$ or $R \cap B = B$.
Indeed, if it is not the case, because of the alternating property, the chosen connected subgraph would contain at least one more vertex colored by $c_b$ than colored by $c_e$, and would not satisfy the multiset constraint.
Therefore, within a \bl, the order of the internal vertices does not matter.

\begin{figure}
\centering
\begin{tikzpicture}[scale=0.84]
\node (v) at (-1.4,3) {$v$};
\node[draw,circle,inner sep=-0.25cm] (c) at (-1,3) {$c_0$};
\foreach \i in {0,2,3,5} {
\foreach \j in {0,...,9} {
\node[draw,circle] (u\i\j) at (\j,\i) {} ;
}
\draw (c) -- (u\i0) -- (u\i1) -- (u\i2) -- (u\i3) -- (u\i4) -- (u\i5) -- (u\i6) -- (u\i7) -- (u\i8) -- (u\i9) ;
\draw[dotted] (u\i9) --++ (0.7,0) ;
}

\draw[very thick] (4.5,-0.3) -- (4.5,0.3) ;
\draw[very thick] (6.5,1.7) -- (6.5,2.3) ;
\draw[very thick] (7.5,2.7) -- (7.5,3.3) ;
\draw[very thick] (5.5,4.7) -- (5.5,5.3) ;
\draw[very thick] (7.5,5.7) -- (7.5,6.3) ;

\node[circle,draw,fill=green] (z1) at (0,2) {} ;
\node[circle,draw,fill=red] (z2) at (6,2) {} ;
\node[circle,draw,fill=green] (z3) at (7,2) {} ;
\node[circle,draw,fill=red] (z4) at (9,2) {} ;

\node[circle,draw,fill=green] (x1) at (0,5) {} ;
\node[circle,draw,fill=red] (x2) at (5,5) {} ;
\node[circle,draw,fill=green] (x3) at (6,5) {} ;
\node[circle,draw,fill=red] (x4) at (9,5) {} ;

\node[circle,draw,fill=green] (y1) at (0,3) {} ;
\node[circle,draw,fill=red] (y2) at (3,3) {} ;
\node[circle,draw,fill=green] (y3) at (4,3) {} ;
\node[circle,draw,fill=red] (y4) at (7,3) {} ;
\node[circle,draw,fill=green] (y5) at (8,3) {} ;

\node[circle,draw,fill=green] (t1) at (0,0) {} ;
\node[circle,draw,fill=red] (t2) at (4,0) {} ;
\node[circle,draw,fill=green] (t3) at (5,0) {} ;

\node (pbe) at (11,6) {$P_{be}$} ;

\draw[thick,decorate,decoration={brace,amplitude=2.5pt}] (10,5.1) -- (10,2.9);
\node (pis) at (11.2,4) {$k$ paths $P_i$s} ;

\draw[thick,decorate,decoration={brace,amplitude=2.5pt}] (10,2.1) -- (10,-0.1);
\node (pijs) at (11.45,1) {${k \choose 2}$ paths $P_{i,j}$s} ;

\foreach \j in {0,2,4,6,8} {
\node[draw,circle,fill=green] (s\j) at (\j,6) {} ;}

\foreach \j in {1,3,5,7,9} {
\node[draw,circle,fill=red] (s\j) at (\j,6) {} ;}

\draw (c) -- (s0) -- (s1) -- (s2) -- (s3) -- (s4) -- (s5) -- (s6) -- (s7) -- (s8) -- (s9) ;
\draw[dotted] (s9) --++(0.7,0) ; 

\node (vd1) at (4.5,1.1) {$\vdots$};
\node (vd2) at (4.5,4.1) {$\vdots$};

\end{tikzpicture}
\caption{Illustration of the global construction and the alternating property. Color $c_b$ is represented in green (light gray) and $c_e$ in red (dark gray). The stopping points just precede the vertical cuts.}
\label{fig:ml-hardness1}
\end{figure}
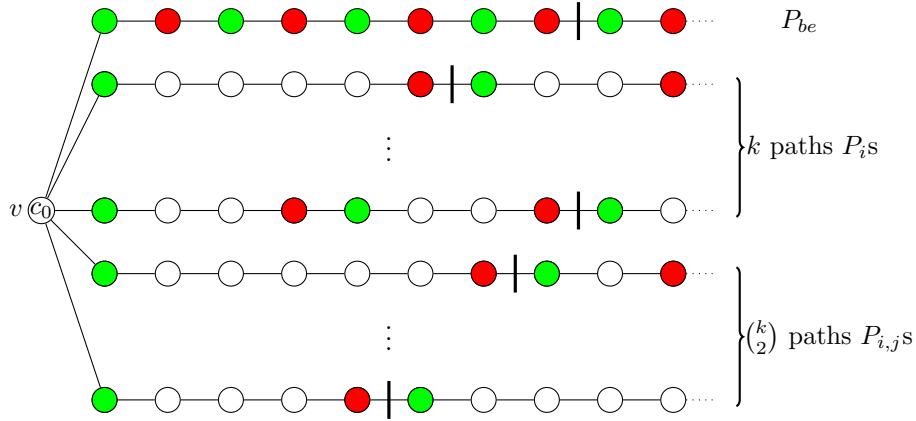

We now describe the path $P_i$ for each $i \in [k]$.
The oriented path $P_i$ consists of $t-1$ copies of the same \bl $B_i$ put one after the other.
The internal vertices of $B_i$ consist of one vertex colored by $li$ for each $l \in [i-1]$ and $t$ vertices colored by $ij$ for each $j \in [i+1,k]$ (see Figure~\ref{fig:ml-hardness2}).
We may recall that the order of the internal vertices of a \bl is irrelevant.
Notice also that the $P_i$s depends only on the number $t$ of vertices per $H_i$. 
As $P_i$ is made of $t-1$ \bls, there are $t$ stopping points, and, intuitively, the $q$-th stopping point corresponds to taking $u_{i,q}$ as part of the multicolored clique in $H$.
As a slight overload of notation, we will also denote by $u_{i,q}$ the $q$-th stopping point of path $P_i$.
By convention, $u_{i,1}$ is $v$.

To motivate the definition of the $P_{i,j}$s, we need to explain how we can think pairs of $H_i \times H_j$ as integers of $[0,t^2-1]$.
Say, the stopping point of a given solution $R$ is $u_{i,q+1}$ in $P_i$ for some $q \in [0,t-1]$, and $u_{j,q'+1}$ in $P_j$ for some  $q' \in [0,t-1]$ (with $i<j$).
The number of vertices colored by $ij$ contained in $R \cap (P_i \cup P_j)$ is $tq+q'$; this number corresponds to a unique pair of stopping points. 
Indeed, function $\phi: x \in [0,t^2-1] \mapsto (\lfloor x/t \rfloor,x \mod t) \in [0,t-1] \times [0,t-1]$ is bijective since $\lfloor x/t \rfloor$ and $x \mod t$ are the quotient and the remainder of the euclidean division of $x$ by $t$.

For any $i<j \in [k]$, the oriented path $P_{i,j}$ consists of $|E_{i,j}|$ \bls whose internal vertices are all colored by $ij$.
We define three auxiliary lists of $|E_{i,j}|$ integers each, indexed from $1$ to $|E_{i,j}|$.
The third list will correspond to how many vertices colored by $ij$ we put in the $|E_{i,j}|$ consecutive blocks. 
The first list $A_{i,j}$ contains, in the increasing order, every integer $x \in [0,t^2-1]$ such that if $\phi(x)=(q,q')$, it holds that $u_{i,q+1}u_{j,q'+1} \in E_{i,j}$.
Intuitively, it is the sorted list of integers in $[0,t^2-1]$ which are \emph{edges} of $E_{i,j}$.
The second list $L_{i,j}$ contains, in the increasing order, all the integers $t^2-x$ such that $x \in A_{i,j}$.
The easiest way to obtain $L_{i,j}$ from $A_{i,j}$ is to complement to $t^2$ each integer in $A_{i,j}$ which yields a list sorted in decreasing order, and to reverse the result.
The third list $D_{i,j}$ is defined by $D_{i,j}[1]:=L_{i,j}[1]$ and for every $h \in [2,|E_{i,j}|]$, $D_{i,j}[k]=L_{i,j}[h]-L_{i,j}[h-1]$.
Finally, for every $h \in [|E_{i,j}|]$, the $h$-th \bl of $P_{i,j}$ gets $D_{i,j}[h]$ vertices colored by $ij$ (see Figure~\ref{fig:ml-hardness2}).
This ends the construction of the instance of \graphm.

\begin{figure}
\centering
\begin{tikzpicture}

\begin{scope}[scale=0.5]

\node (Pi) at (-1,0) {$P_i$} ;
\foreach \t in {0,1,2}{

\begin{scope}[xshift=7*\t cm]
\node[draw,circle,fill=green] (a0\t) at (0,0) {};
\node[draw,circle,fill=red] (a5\t) at (6,0) {};


\foreach \i in {1,...,4}{
         \node[inner sep = -0.3cm,draw,circle] (a\i\t) at ({\i+0.5},0) {$ij$}; 
}       
\draw[dotted] (a0\t) -- (a1\t) ;
\draw (a1\t) -- (a2\t) ;
\draw[dotted] (a2\t) -- (a3\t) ;
\draw (a3\t) -- (a4\t) ;
\draw[dotted] (a4\t) -- (a5\t) ;

\node (t\t) at (3,-1) {$t$} ;

\draw[decorate,decoration={brace,amplitude=2.5pt}] (4.8,-0.5) -- (1.2,-0.5);

\end{scope}
}
\draw(a50) -- (a01) ;
\draw(a51) -- (a02) ;

\node (Pj) at (-1,-3) {$P_j$} ;

\foreach \t in {0,1,2}{

\begin{scope}[xshift=7*\t cm,yshift=-3cm]
\node[draw,circle,fill=green] (b0\t) at (0,0) {};
\node[draw,circle,inner sep = -0.3cm,] (b1\t) at (3,0) {$ij$};
\node[draw,circle,fill=red] (b2\t) at (6,0) {};
 

\draw[dotted] (b0\t) -- (b1\t) -- (b2\t) ;
\end{scope}
}
\draw(b20) -- (b01) ;
\draw(b21) -- (b02) ;

\draw[dotted] (a52) --++(1,0) ;
\draw[dotted] (b22) --++(1,0) ;
\end{scope}

\node (uj1) at (-0.5,-2) {$u_{j,1}$};
\node (uj2) at (3,-2) {$u_{j,2}$};
\node (uj3) at (6.5,-2) {$u_{j,3}$};
\node (uj4) at (10,-2) {$u_{j,4}$};

\node (ui1) at (-0.5,-0.5) {$u_{i,1}$};
\node (ui2) at (3,-0.5) {$u_{i,2}$};
\node (ui3) at (6.5,-0.5) {$u_{i,3}$};
\node (ui4) at (10,-0.5) {$u_{i,4}$};

\begin{scope}[yshift=-3cm]
\node (pij) at (-0.5,0) {$P_{i,j}$} ;
\node[draw,circle,fill=green] (v1) at (0,0) {} ;
\node[draw,rectangle] (v2) at (1.25,0) {$D_{i,j}[1] \times ij$} ;
\node[draw,circle,fill=red] (v3) at (2.5,0) {} ;

\node[draw,circle,fill=green] (v4) at (3,0) {} ;
\node[draw,rectangle] (v5) at (4.25,0) {$D_{i,j}[2] \times ij$} ;
\node[draw,circle,fill=red] (v6) at (5.5,0) {} ;

\node[draw,circle,fill=green] (v7) at (6,0) {} ;
\node[draw,rectangle] (v8) at (7.25,0) {$D_{i,j}[3] \times ij$} ;
\node[draw,circle,fill=red] (v9) at (8.5,0) {} ;

\draw (v1) -- (v2) -- (v3) -- (v4) -- (v5) -- (v6) -- (v7) -- (v8) -- (v9) ;
\draw[dotted] (v9) --++(1,0) ;
\end{scope}

\end{tikzpicture}
\caption{The oriented paths $P_i$, $P_j$, and $P_{i,j}$. Again, color $c_b$ is represented in green (light gray) and color $c_e$ in red (dark gray). Note that the $P_i$s do depend only on the number $t$ of vertices per color class, while $P_{i,j}$ actually encodes the adjacency between $H_i$ and $H_j$ in some \emph{flattened} form.}
\label{fig:ml-hardness2}
\end{figure}
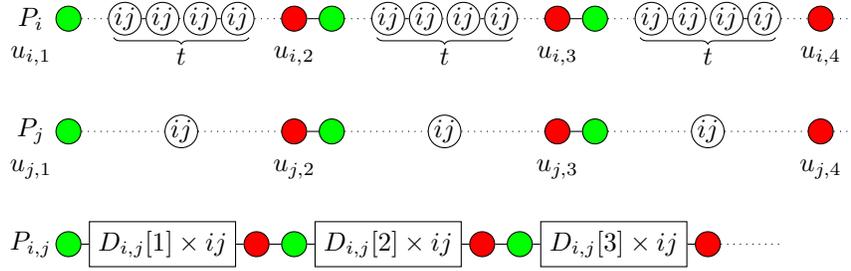

Suppose there is a multicolored clique $C:=\{u_{1,q_1},\ldots,u_{k,q_k}\}$ in $H$.
We construct a solution $R$ to the produced instance $(T,c,M)$ in the following way.
For each $i \in [k]$, the stopping point of $R$ in path $P_i$ is $u_{i,q_i}$.
For any pair $i<j \in [k]$, let $y_{i,j}:=t^2-\phi^{-1}(q_i-1,q_j-1)$, and let $h_{i,j}$ be the index such that $y_{i,j}=L_{i,j}[h_{i,j}]$.
The stopping point of $R$ in path $P_{i,j}$ is right after its $h_{i,j}$-th block.
The subtree induced by those $k+{k \choose 2}$ stopping points contains the same number $z$ of vertices colored by $c_b$ and of vertices colored by $c_e$.
As $z$ is non-negative and cannot exceed $s$, solution $R$ can and will stop after $s-z$ \bls in $P_{be}$, thereby fulfilling the multiset constraint for colors $c_b$ and $c_e$.
By construction (from vertex $v$ along the oriented paths), $R$ induces a connected subgraph.

What remains to be seen is that $h_{i,j}$ is well defined and that, for each $i<j \in [k]$, $R$ contains exactly $t^2$ vertices colored by $ij$.
A preliminary easy observation is that vertices colored by $ij$ only appear in three oriented paths: $P_i$, $P_j$ and $P_{i,j}$.
For any pair $i<j \in [k]$, as $C$ is a clique, $u_{i,q_i}u_{j,q_j} \in E_{i,j}$.
Thus, the value $\phi^{-1}(q_i-1,q_j-1)$ is in $A_{i,j}$, and so, $y_{i,j}=t^2-\phi^{-1}(q_i-1,q_j-1)$ is in $L_{i,j}$.
This means that $h_{i,j}$ exists.
Also, by definition of $\phi$, $\phi^{-1}(q_i-1,q_j-1)$ corresponds to the number of vertices colored by $ij$ in $R \cap (P_i \cup P_j)$.
Therefore, $y_{i,j}$ is exactly the number of vertices colored by $ij$ we want to have in $R \cap P_{i,j}$.
As we stop $R$ in $P_{i,j}$ after $h_{i,j}$ \bls, the number of vertices colored by $ij$ in $R \cap P_{i,j}$ is $\Sigma_{1 \leqslant r \leqslant h_{i,j}}D_{i,j}[r]$.
And, $\Sigma_{1 \leqslant r \leqslant h_{i,j}}D_{i,j}[r]=(\Sigma_{2 \leqslant r \leqslant h_{i,j}}L_{i,j}[r]-L_{i,j}[r-1])+L_{i,j}[1]=L_{i,j}[h_{i,j}]=y_{i,j}$.
Hence, the total number of vertices colored by $ij$ in $R$ is $\phi^{-1}(q_i-1,q_j-1)+y_{i,j}=t^2$.

Now, suppose that there is no multicolored clique in $H$.
We will show that there cannot be a solution to the instance of \graphm.
For the sake of contradiction, we assume that $R$ is a solution.
As explained during the construction, vertex $v$ has to be in $R$ and the stopping points in each oriented path $P_i$, $P_{i,j}$, and $P_{b,e}$ should coincide with the end of \bls.
In particular, in each $P_i$, the stopping point of $R$ should be a vertex $u_{i,q}$.
Thus, let $u_{1,q_1},\ldots,u_{k,q_k}$ be the stopping points of $R$ in $P_1,\ldots,P_k$.
As there is no multicolored clique in $H$, there exists at least one pair $i<j \in [k]$, such that $u_{i,q_i}u_{j,q_j} \notin E_{i,j}$.
Let $h$ be the number of \bls in $R \cap P_{i,j}$; in other words, $R$ stops in $P_{i,j}$ after $h$ blocks.
We now show that $R$ cannot contain exactly $t^2$ vertices colored by $ij$, and hence, is not a solution.
The number of vertices colored by $ij$ in $R \cap (P_i \cup P_j)$ is $\phi^{-1}(q_i-1,q_j-1) \notin A_{i,j}$.
As $x \in [0,t^2-1] \mapsto t^2-x \in [t^2]$ is bijective, it means that $t^2-\phi^{-1}(q_i-1,q_j-1) \notin L_{i,j}$.
Besides, the number of vertices colored by $ij$ in $R \cap P_{i,j}$ is $\Sigma_{1 \leqslant r \leqslant h}D_{i,j}[r]$.
We observed in the previous paragraph that $L_{i,j}[h]=\Sigma_{1 \leqslant r \leqslant h}D_{i,j}[r]$.
Hence $t^2-\phi^{-1}(q_i-1,q_j-1) \neq \Sigma_{1 \leqslant r \leqslant h}D_{i,j}[r]$, so $\phi^{-1}(q_i-1,q_j-1) + \Sigma_{1 \leqslant r \leqslant h}D_{i,j}[r] \neq t^2$.
\end{proof}
}

As it is usually the case with FPT reductions from \textsc{Multicolored $k$-Clique} using edge representations the parameter goes from $k$ to $\Theta(k^2)$.
Thus, concerning running-time lower bounds, the previous reduction only shows that solving \graphm in time $n^{o(\sqrt{\text{ml}(G)+|\mathcal C|})}$ would also solve \textsc{Multicolored $k$-Clique} in time $n^{o(k)}$ which is known to disprove ETH, and even imply that $\fpt=\wone$ \cite{Chen05}.
Nevertheless, we can strengthen this lower bound by performing the same reduction from \textsc{Partitioned Subgraph Isomorphism}.
In the \textsc{Partitioned Subgraph Isomorphism} problem, one is given two graphs $H$ and $G$.
The vertices of graph $H$ are partitioned into $|V(G)|$ classes $C_v$ one for each vertex $v$ of $G$.
The goal is to find an injective mapping $h : V(G) \rightarrow V(H)$ such that if $uv \in E(G)$, then $h(u)h(v) \in E(H)$, and for each $v \in V(G)$, $h(v) \in C_v$.
Under ETH, \textsc{Partitioned Subgraph Isomorphism} cannot be solved in time $n^{o(k / \log k)}$ where $k$ is the number of edges of the smaller graph $G$ \cite{Marx10}.
Observe that we can ignore isolated vertices in $G$ (we are looking for a subgraph \emph{not} an induced subgraph).
Thus, the number of edges in $G$ is at least $|V(G)|/2$, and ETH even implies that \textsc{Partitioned Subgraph Isomorphism} cannot be solved in time $n^{o(k / \log k)}$ where $k=|V(G)|+|E(G)|$.

The reduction from \graphm to \textsc{Partitioned Subgraph Isomorphism} encode the graph $H$ partitioned into the $C_v$s but only introduce a color $ij$ and a path $P_{i,j}$ if there is an edge in $G$ between the $i$-th and the $j$-th vertex.
The number of leaves in $T$ is $|V(G)|+|E(G)|+1$ and the number of colors of $\mathcal C$ is $|E(G)|+3$.
Thus, we get that, under ETH, \graphm cannot be solved in $n^{o((\text{ml}(G)+|\mathcal C|) / \log{(\text{ml}(G)+|\mathcal C|)})}$.
Therefore, our algorithm running in time $n^{O(\text{ml}(G))}$ is probably optimal up to logarithmic factors in the exponent.

The \graphm problem on subdivisions of stars can be reformulated as the following problems on words: given a set of $k+1$ words $w_1, \ldots, w_k$, and $w$ over an alphabet $\Sigma$, find $w'_1, \ldots, w'_k$, such that for each $i \in [k]$, $w'_i$ is a prefix of $w_i$, and the concatenation $w'_1w'_2 \ldots w'_k$ is an anagram of $w$.
Indeed, hard instances of \graphm on subdivisions of stars are such that the center of the subdivided star should necessarily be in a solution (otherwise, the whole solution is entirely contained in an induced path, and can be computed in polynomial time).
Then, letters correspond to colors, $w$ to the multiset $M$, and the $w_i$'s to the words formed by the colors of the vertices in each oriented path.
Therefore, Theorem~\ref{thm:max-leaf-hard} entails that this problem is $\wone$-hard parameterized by $k+|\Sigma|$ (number of words plus size of the alphabet). 
However, 
 as far as we know, this problem has not appeared in the literature.

We may finally observe that \graphm on paths \emph{is} an established string problem going by the name of \emph{jumbled pattern matching} (see for instance \cite{Burcsi12}).
In this problem, one has to find, given a string and a Parikh vector (or multiset of letters), a substring whose occurences of letters match the Parikh vector.
Therefore, \graphm can be seen as a generalization of this string problem to more complex structures. 

\section{Conclusion and open problems}\label{sec:conclusion}
Figure~\ref{fig:recap} sums up the parameterized complexity landscape of \graphm with respect to structural parameters. 
For parameter maximum independent set 
 the complexity status of \graphm remains unknown. 
Even when the problem is in $\fpt$, polynomial kernels tend to be unlikely;
be it for the natural parameter even on comb graphs~\cite{Ambalath2010} or for the vertex cover number or the distance to clique (Theorem~\ref{th:nokernelvc}). 
Is it also the case for parameter cluster editing number?

On the one hand, we saw that our algorithm running in $O^*(3^k)$ for parameter distance to clique is probably close to optimal, since $O^*((2-\varepsilon)^k)$ is unlikely. 
On the other hand, for parameter vertex cover number, for instance, we have a larger room for improvement between the $2^{O(k \log k)}$-upper bound and the $2^{o(k)}$-lower bound under ETH.
Can we improve the algorithm to time $2^{O(k)}$, or, on the contrary, show a stronger lower bound of $2^{o(k \log k)}$ (potentially with the framework developed by Lokshtanov \emph{et al.}~\cite{Lokshtanov11})?

A possible future work would be to see if the FPT algorithms presented in the article can be extended to the more general \textsc{List Graph Motif}, where a vertex can choose its color among a private list of colors, without damaging too much their running time.

Finally, one could consider more restricted versions (when, for instance, the number of colors, or the maximum multiplicity of the motif, or the maximum number of occurences of a color in the graph, is bounded). 
This line of work is sometimes called \emph{multi-parameter analysis}, where one seeks for FPT algorithms with respect to subset of parameters. 
Let us recall, as an example, that \graphm is in $\xp$ if the parameter is the treewidth of the graph plus the number of colors in the motif~\cite{fellows2011}.






\section*{Acknowledgments}
The work of the first author is supported by the European Research Council (ERC) grant "PARAMTIGHT: Parameterized complexity and the search for tight complexity results," reference 280152.

\bibliographystyle{abbrv}
\bibliography{structural}


\end{document}